\documentclass[twocolumn]{article}
\usepackage{graphicx} % for EPS, load graphicx instead 
\usepackage{color}
\pdfoutput=1

\usepackage{xspace}
\usepackage{enumitem}
\usepackage{eqnarray}
\usepackage{amsthm}
\usepackage{amsmath}
\newtheorem{theorem}{Theorem}
\newcommand{\kic}{{\sc TAS}\xspace} 
\newcommand{\nn}{{none}\xspace} 
\usepackage[noend]{algpseudocode}
\usepackage[nothing]{algorithm}

\algrenewcommand\algorithmicrequire{\textbf{Input:}}
\algrenewcommand\algorithmicensure{\textbf{Output:}}
\algnewcommand\algorithmicpass{\textbf{pass}}
\algnewcommand\algorithmicbreak{\textbf{break}}
\algnewcommand\algorithmicnot{\textbf{not}}

\begin{document}

\title{It's about time: Online Macrotask Sequencing in Expert Crowdsourcing}
%maybe we should keep the above title for the journal and post it under a different name in arXiv
%Previous title: Online Task Assignment and Sequencing in Expert Crowdsourcing}
%\ALTERNATIVE TITLES: Task Assignment and Scheduling in Expert Crowdsourcing // Task Assignment and Scheduling Optimization in Expert Crowdsourcing // It's about time: Scheduling Optimization for Macrotasks in Expert Crowdsourcing//It's about time: Optimizing Macrotask Assignment and Sequencing in Expert Crowdsourcing// It's about time: Online Macrotask Sequencing in Expert Crowdsourcing//  

\author{
Heinz Schmitz\\
Trier University of Applied Sciences, Germany\\
{\tt \small h.schmitz@hochschule-trier.de}\\
Ioanna Lykourentzou\\
Luxembourg Institute of Science and Technology\\
{\tt \small ioanna.lykourentzou@list.lu}
}

\maketitle
\begin{abstract}
\begin{quote}
We introduce the problem of Task Assignment and Sequencing (TAS), which adds the timeline perspective to expert crowdsourcing optimization. Expert crowdsourcing involves macrotasks, like document writing, product design, or web development, which take more time than typical binary microtasks, require expert skills, assume varying degrees of knowledge over a topic, and require crowd workers to build on each other's contributions. Current works usually assume offline optimization models, which consider worker and task arrivals known and do not take into account the element of time. Realistically however, time is critical: tasks have deadlines, expert workers are available only at specific time slots, and worker/task arrivals are not known a-priori. Our work is the first to address the problem of optimal task sequencing for online, heterogeneous, time-constrained macrotasks. We propose {\sc tas-online}, an online algorithm that aims to complete as many tasks as possible within budget, required quality and a given timeline, without future input information regarding job release dates or worker availabilities. Results, comparing {\sc tas-online} to four typical benchmarks, show that it achieves more completed jobs, lower flow times and higher job quality. This work has practical implications for improving the Quality of Service of current crowdsourcing platforms, allowing them to offer cost, quality and time %guarantees 
improvements for expert tasks.
%We introduce the problem of crowdsourcing Task Assignment and Sequencing (TAS), which adds the timeline perspective to expert crowdsourcing optimization. Subsequently we propose {\sc tas-online}, an online algorithm that aims to complete as many tasks as possible within budget and required quality without future input information regarding job release dates or worker availabilities. Results, comparing {\sc tas-online} to four typical benchmarks, show that it can achieve more completed jobs, lower flow times and higher job quality. This work has practical implications for improving the outcome of expert crowdsourcing with cost, quality and time constraints. ~\footnote{\textcolor{red}{Can I use \textbf{scheduling} instead of \textbf{sequencing} in the paper title and introduction? If I understood well sequencing is a sub-case of scheduling, where the tasks have sequence-dependent setups. Can we specify that we work on this particular sub-case later on, when we define the problem model?}}
\end{quote}
\end{abstract}

%%%%%%%%%%%%%%%%%%%%%%%%%%%%%%%%%%%%%%%%%%%%%%%%%%%%%%%%%%%%%%%%%%%%%%%%
%%%%%%%%%%%%%%%%%%%%%%%%%%%%%%%%%%%%%%%%%%%%%%%%%%%%%%%%%%%%%%%%%%%%%%%%
\section{Introduction}
\label{introduction}

As the appeal of crowd work increases, there is a growing need to provide support for more complex tasks and workflows. Examples of such tasks include document editing, product design, social innovation and idea development, offered through dedicated platforms (upWorks\footnote{https://www.upwork.com/}, crowdSpring\footnote{http://www.crowdspring.com/} etc.), or incorporated into traditional ones through complex workflows (e.g. the recently launched CrowdFlower Labs\footnote{http://www.crowdflower.com/blog/introducing-crowdflower-labs}). This type of crowdsourcing is often referred to as \textit{expert crowdsourcing}, and the tasks that it involves are referred to as \textit{macrotasks}~\cite{Haas:2015:AMC:2824032.2824062}. Macrotasks differ from the typically crowdsourced microtasks in that they require expert skills, assume varying degrees of knowledge over a topic, may take more worker time and often involve task dependency, i.e. workers building on each other's contributions. 

Together with the demand for complex tasks and their supporting workflows, customers are increasingly interested in performance guarantees, i.e. the \textit{optimization} of expert crowdsourcing in terms of cost, quality and timeliness. Recent studies that examine expert crowdsourcing optimization~\cite{roy2015,Goel:2014:ATW:2567948.2577311} typically seek to find worker assignments per task such that the worker contributions add up to a required quality threshold within a given budget. Roughly speaking, the crowdsourced tasks play the roles of multiple knapsacks with some additional concepts like domain-specific expertise and wages per worker, different models of acceptance probabilities or types of quality aggregation. Unfortunately current studies do not take into account the aspect of time, for example in terms of task deadlines, worker time constraints, or time-dependent worker/task variability. As such these studies examine only the \textit{assignment part} of the worker allocation
%ALTERNATIVELY: optimization 
problem (finding which worker should take which task), but not the \textit{sequencing part} (identifying when each worker should contribute). Moreover they usually assume an offline setting, where the algorithms are provided with the complete worker/task input information at once. \textit{In a realistic crowdsourcing setting however, time is an inherent property}: 
customers require the tasks to finish upon a certain deadline, expert workers are available only at specific time slots, and worker/task arrival or departure information is not a-priori known. Optimizing for time is thus crucial, and raises the need not only for worker-to-task assignment but also for sequencing. It also raises the need for online rather than index-based algorithms, which can take efficient sequencing decisions having access only to time-dependent information that is available until their decision point.

In this paper we introduce the problem of crowdsourcing Task Assignment and Sequencing (TAS), which adds the \emph{timeline perspective} to the crowdsourcing allocation optimization model: How can we find task assignments that can be rolled out in a realistic timeline, featuring unknown task release dates and worker availabilities, as it is the case for real platforms? 

To the best of our knowledge, this is the first work that addresses the problem of assignment \textit{and} sequencing optimization for expert crowdsourcing tasks. Overall, our three main contributions with this paper are:
\begin{itemize}
\item We explicitly add the \emph{timeline perspective} to task assignment modeling in expert crowdsourcing. That is, our models include not only the worker-to-task-assignments, but also the rolling out of these assignments along a timeline under reasonable constraints. We call this problem modeling TAS and prove its strong NP hardness. 
\item We propose a \emph{online algorithm}, {\sc tas-online}, which seeks to complete as many jobs as possible within budget, required quality and given timeline, by computing worker sequence-to-task matchings, and without any future input information regarding job release dates or worker availabilities.
\item We illustrate, through simulated and real-world experiments, that {\sc tas-online} can achieve more completed jobs, lower flow times and higher quality compared to four typical benchmarks.
\end{itemize}

The rest of this paper is organized as follows. In section~\ref{related} we recapitulate the related literature on crowdsourcing optimization, starting from earlier works that focus on micro-tasks and reaching latest research efforts on knowledge-intensive macro-tasks. In section~\ref{tas} we describe the characteristics of the expert crowdsourcing setting that this work applies and illustrate, through an example, why taking time into account matters in this particular setting. In this section we also formally model the TAS problem and prove its strong NP-hardness. Next, in section~\ref{algorithm} we describe the proposed online algorithm ({\sc tas-online}) for the solution of the TAS problem. In section~\ref{results}, we present and discuss the experimental results, obtained on both simulated and real-world data. These results compare {\sc tas-online} with four benchmarks found in the literature, and show that {\sc tas-online} achieves higher numbers of completed jobs (both in terms of absolute value and as a percentage of the solution's upper bound), lower flow times (the time between a job's release date and the latest assignment on that job), better budget utilization and higher levels of quality, comparable only the respective {\sc tas-offline} version for certain of the above measures. Finally, we discuss possible extensions of the TAS model and algorithm (section~\ref{extensions}) and end with the paper's main findings and conclusions (section~\ref{conclusion}).
\vspace{-1em}

\section{Related Work}
\label{related}

\subsection{Crowdsourcing Optimization}
Crowdsourcing optimization is a term used in various problem settings, including optimizing the selection of worker labor channels to improve performance~\cite{DBLP:conf/hcomp/KaranamCCDR14}, discovering the optimal worker wage  \cite{Horton:2010:LEP:1807342.1807376}, determining the optimal number of workers to undertake each task so as to maximize quality and minimize cost (a method referred to as ``plurality optimization" and applicable on n-ary tasks with an objective ``true value")~\cite{Mo2013}, 
%optimizing the worker queue size (i.e. determining the optimal number of workers to queue so as to maximize task response time with the minimum cost)~\cite{journals/corr/abs-1204-2995}, --> Has been moved below
or identifying the optimal set of tasks to forward to the crowd (for systems like database query execution ones, which are based partially on crowdsourcing and partially on automated methods)~\cite{7052378}.

The family of optimization problems that our work falls into is \textit{allocation optimization},
%ALTERNATIVELY: allocation optimization
i.e. the identification of which worker should work on which task and when, in order to optimize one or more global performance metrics, which usually include cost, quality and number of acceptable tasks. This family of problems 
%assumes a task push model (workers are recommended specific tasks) rather than a pull model (workers select the tasks they will work on). It 
consists of two distinct optimization problems, \textit{task assignment} and \textit{task sequencing}. Task assignment examines which worker should be given which task. Task sequencing adds the element of time, examining \textit{when} will the worker be given the task. Most existing works focus on the first problem, i.e. task assignment, either for microtasks or for macrotasks. \textit{Microtasks} are tasks that require a small amount of worker time, accept binary (true/false) or n-ary (multiple choice) worker inputs, and the quality of which is determined through methods such as majority voting (assigning multiple workers per micro-task). \textit{Macrotasks}~\cite{Haas:2015:AMC:2824032.2824062} require more worker time, accept open-ended continuous (rather than binary or n-ary) worker inputs and their quality is determined through external subjective evaluation (for example peer review).

\subsection{Optimizing Task Assignments}
Regarding \textit{microtask assignment optimization}, Karger et al.~\cite{kargerBudget} work with homogeneous microtasks (that  all have the same level of difficulty and do no distinguish among task ``topics"), and propose a matching algorithm inspired by the standard belief propagation algorithm for approximating max marginals, which is order-optimal and minimizes cost. This study is among the first to show that the problem of task matching in crowdsourcing can be transformed to a bipartite graph design problem, where workers are one part of the graph, tasks are the other and the edges represent assignments of workers to tasks. Ho et al.~\cite{DBLP:conf/aaai/HoV12} work with heterogeneous microtasks of n-ary classification quality on a model where worker skills per microtask are a-priori unknown, and propose a two phase exploration-exploitation assignment algorithm that seeks to maximize the total benefit of the requester and is competitive with respect to its counterpart of known worker skills. Yuen et al.~\cite{Yuen:2012:TRC:2442657.2442661,Yuen:2012:TPM:2428413.2428477,Yuen:2011:TMC:2085036.2085206} propose a matrix factorization approach that utilizes the workers' task performance and search history to derive their  preferences and perform an improved task-to-worker matching. 
Regarding \textit{macrotask assignment optimization} Goel et al.~\cite{Goel:2014:ATW:2567948.2577311} and Roy et al.~\cite{roy2015,DBLP:journals/corr/RoyLTAD14} both propose task-to-worker assignment optimizations (the first using a mechanism design-based approach and the second through an index-based approach) on models that consider heterogeneous macrotasks and where the optimization goal is to maximize the utility of the requester while ensuring budget feasibility. Yue et al.~\cite{yue2015} add to this model the element of team instead of individual worker assignments, and propose a heuristic genetic algorithm that optimizes for task budget and quality, taking into account worker pay expectations, skills and availability.  Jabbari et al.~\cite{jabbari2015} add another interesting facet to the heterogeneous task assignment model, by extending it to cover the online aspect of the problem (workers arrive online, no prior knowledge over arrivals), and they impose certain constraints that must be respected, such as declaring feasible tasks that workers can handle and the payment they require. The difference of the above works with ours is that their models do not consider the element of time, i.e., they focus only on the task assignment and not the task sequencing aspect of the problem. 

%The above models are offline (algorithm has a-priori knowledge of tasks and workers) and time-invariant (algorithm does not need to respect task deadlines or a specific timeline). Hence the algorithms have a-priori access to all tasks to be allocated and the available workers, while they do not need to respect task deadlines or a specific timeline. 
%can eventually obtain almost perfect information and adapt its strategy
%Regarding time variant

%for macro-tasks, i.e. tasks that require skill diversity, more worker time than micro-tasks and do not have an n-ary but a continuous quality, Macrotaks--> value~\cite{Haas:2015:AMC:2824032.2824062}. 

\subsection{Optimizing Task Timing}
%TIME OPTIMIZATION 
Finally, a few studies focus precisely on \textit{time-sensitive optimization}. 
%FOR BINARY TASKS
Regarding \textit{time-sensitive microtasks}, Yu et al.~\cite{6782912} optimize the number of tasks to recommend to each worker per time unit with the objective of maximizing the time average number of successful (i.e. of acceptable quality) jobs for a given time period. Their model assumes binary task quality, a pull-and-filter task selection model (workers select which tasks they want to work on and the system filters these selections according to its optimization objective) and performs task allocation 
%ALTERNATIVELY: scheduling
on the basis of worker accuracy measured in a [0,1] scale using a heuristic algorithm of linearithmic complexity. Although this study does incorporate the element of time, it is different than ours in that the model it uses assumes binary, homogeneous tasks rather than heterogeneous tasks of continuous quality.  The use of homogeneous tasks (all tasks have the same difficulty, no distinction of task topics) means that optimization needs to be performed in terms of the number of jobs per worker, rather in terms of allocating
%ALTERNATIVELY: scheduling
 specific workers to specific jobs according to their skills. Bernstein et al.~\cite{journals/corr/abs-1204-2995,Bernstein:2011:CTS:2047196.2047201} also work with homogeneous microtasks and propose a retainer model for pre-recruiting (reserving) the optimal number of workers, so as to minimize task completion latency.
%of on-demand workers, thus minimizing the latency of task completion. 
%optimizing the worker queue size (i.e. determining the optimal number of workers to queue so as to maximize task response time with the minimum cost)~\cite{journals/corr/abs-1204-2995}, 
This work however does not take into account worker skills and subsequently does not seek to maximize task quality. On heterogeneous microtasks, Faridani et al.~\cite{conf/aaai/FaradaniHI11} and 
Minder et al.~\cite{minder2012} add the element of pricing to the problem model, proposing task pricing algorithms that aim to maximize the number of tasks finishing on time (the first study), while respecting budget and quality constraints (the second study). Mechanism design~\cite{Nath:2012:MDT:2436756.2436773, rajan2013} and multi-armed bandit mechanism design~\cite{Biswas:2015:TBF:2772879.2773291} mechanisms have also been employed to manipulate the time behavior of the crowd towards an efficient execution of time-critical tasks. The above works are different from our work, in that they do not explicitly sequence the tasks to the workers but they rather seek to incentivize the crowd's timely responses in order to achieve the time-related task objective.

%For Macro-tasks
In regards to \textit{time-sensitive macrotasks} Khazankin et al.~\cite{6450935} propose a mathematical optimization approach that learns the task selection behavior of workers and then executes tasks in a manner that optimizes for cost and considers deadlines. This approach does not consider task quality, yet it is one of the first attempts to sequence time-sensitive macrotasks. Finally, Boutsis and Kalogeraki ~\cite{6888877} propose an multi-objective optimization approach which searches for Pareto-optimal solutions, seeking to identify the group of workers (among multiple candidate groups) with the highest probability of finishing the task on time. This approach is different than ours, in that it does not apply sequencing along a timeline, but rather makes one-shot 
%one-time% 
assignments based on the worker probabilities of meeting the deadline. 

Overall, crowdsourcing optimization studies have so far examined mainly the assignment but not the sequencing aspect of the problem. The works that optimize for time-sensitive task characteristics are few and they either focus on binary/n-ary microtasks (which differ significantly from the expert macrotasks that we are interested in) or they do not sequence the worker-to-task assignments along a timeline. 
%Our work is the first to address the problem of optimal task sequencing for \textit{online}, \textit{heterogeneous},  
We address in our work the problem of optimal task sequencing for \textit{online}, \textit{heterogeneous},  
%(many topics), iii) 
\textit{time-constrained}, 
%(requiring different worker skills), iv) 
\textit{macrotasks}. 
%, like those seen in expert crowdsourcing. 
% (tasks that require open-ended, continuous-quality worker contributions rather than binary or n-ary ones, and which take up more worker time compared to binary/n-ary microtasks).   
As we will see in the following section, it is this type of tasks that expert crowsourcing consists of, and thus their optimization has significant impact on many recent platforms and applications.

\section{Task Assignment and Sequencing (\kic)}
\label{tas}
In this section we first describe, from a high-level viewpoint, the expert crowdsourcing task model which we target through this work. Then we provide an example to illustrate the importance of adding the timeline sequencing element into the 
%optimization problem of the 
above setting. Next we define the \kic problem model, in terms of the input data, feasible solutions, constraints and optimization goal. Finally we analyze this problem model in terms of complexity. 
\subsection{Expert Crowdsourcing Setting}
The expert crowdsourcing problem setting, at which this work is aimed, features some very particular characteristics that make it unique compared to other crowdsourcing problem settings:
\begin{enumerate}
\item \textbf{Heterogeneous rather than homogeneous tasks}. We work with crowdsourcing tasks that require different skills and skill levels from the workers, and that belong to multiple ``topics" (rather than a single one). Workers in this setting possess a set of skills, and are less replaceable and less abundant than crowdsourcing settings that consider homogeneous tasks and skills (everyone can do every task).

\item \textbf{Macro- rather than microtasks.} Whereas microtasks accept binary or n-ary (multiple choice) worker input, and their quality is defined by assigning multiple simultaneous workers on the task and then performing majority voting, macrotasks feature open-ended worker input (e.g. write a product description), and their quality is built by one worker contribution after the other (sequentially rather than simultaneously). 

\item \textbf{Online rather than offline.} Rather than working on a simplified offline setting, where the pool of workers and/or tasks are a-priori known, we consider an online setting, where workers demonstrate a dynamic flow of arrivals and departures and tasks arrive in an unpredictable manner. Any sequencing decision must be made based on task/worker information available up to the specific point in time.
%As analyzed above, most related works consider a simplified off-line setting, where the pool of workers and/or tasks are a-priori known. In reality however, the crowdsourcing platform does not have the luxury of seeing all the workers and then selecting the right ones, because workers commonly work on multiple platforms and demonstrate a very dynamic flow. Similarly, the platform does not have the luxury to know beforehand all the tasks that will be crowdsourced, but it rather receives them as they are provided by the requesters. 

\item \textbf{Time-constrained rather than only cost/quality-constrained tasks.} In addition to the need of achieving a certain utility metric (e.g. quality, number of acceptable tasks etc.) and the need to keep costs within budget, the online macrotasks of this setting have a deadline, i.e., they must also finish by a specific time point.
%Most works on heterogeneous, online macrotasks optimize for a certain metric of utility (e.g. quality, number of acceptable tasks etc.) while keeping tasks within budget. Very few consider the element of time, e.g. that the tasks must finish on a specific deadline.
\end{enumerate}

\paragraph{Application areas.}
Many tasks and platforms, especially recent ones, fit the above \textit{expert crowdsourcing} setting and could benefit from its optimization. The first example are platforms such as upWork\footnote{https://www.upwork.com/} that work with freelancer experts on creative tasks such as web design and development, document writing, or coding. Social innovation platforms such as OpenIDEO\footnote{https://openideo.com/} or creative product design platforms like Quirky\footnote{https://www.quirky.com/}, where users build on one another's ideas, could also directly benefit from the optimization of the above setting. Finally collaborative document editing applications, such as corporate wikis~\cite{Lykourentzou201018}, where it is possible to sequence worker contributions along a timeline could significantly improve from our approach.

\subsection{The importance of the time element}
Before giving the formal definition of the \kic optimization problem, we illustrate through an example the importance of adding the perspective of time, and 
%how this addition has a significant impact on optimal task sequencing for expert crowdsourcing.
how this addition has a significant on performance in expert crowdsourcing.

\paragraph{Example.} Suppose there are only two jobs given, both from the same knowledge domain. Each job $j = (Q,C)$ has a quality threshold $Q$ that needs to be reached, and a cost threshold $C$
that must not be exceeded. For this example suppose
$$
j_0 = (5,5) \text{~~and~~} j_1=(4,4).
$$
On the other hand, each worker $i=(e,w)$ has an expertise $e$ that increases the quality of a job, and a wage $w$ that consumes this job's budget.
Let us assume that three workers
$$
i_0 = (2,3), i_1=(3,2) \text{~~and~~} i_3=(2,1)
$$
are given. Then each job has two possible assignments within budget and with sufficient quality:
\begin{eqnarray*}
j_0 &:& \{i_0,i_1\} \text{~~or~~} \{i_1,i_2\}\\
j_1 &:& \{i_0,i_2\} \text{~~or~~} \{i_1,i_2\}
\end{eqnarray*}
For both jobs the latter assignment seems to be preferable over the other since workers $\{i_1,i_2\}$ provide more quality for less cost.
Now we look at a sequencing period of three timeslots with limited worker availability as follows (both jobs are released immediately):
\begin{center}
\begin{tabular}{rccc}
  timeslot & $0$ & $1$ & $2$  \\ \hline
  $i_0$ & ~ & ~ & $\times$ \\
  $i_1$ & ~ & $\times$ & ~ \\
  $i_2$ & $\times$ & ~ & $\times$ \\
\end{tabular}
\end{center}
Since worker $i_1$ is available only at a single timeslot it is clear that at most one job can realize the preferable assignment mentioned above.
So assume for the moment that we choose modestly $j_0 \leftarrow \{i_0,i_1\}$ for $j_0$.
This gives us the following partial schedule for the workers:

\begin{center}
\begin{tabular}{rccc}
  timeslot & $0$ & $1$ & $2$  \\ \hline
  $i_0$ & ~ & ~ & $j_0$ \\
  $i_1$ & ~ & $j_0$ & ~ \\
  $i_2$ & $\times$ & ~ & $\times$ \\
\end{tabular}
\end{center}
But now none of the two feasible assignments for $j_1$ can be realized since only worker $i_2$ remains available. 
Although the assignment $j_1 \leftarrow \{i_2\}$ is within budget, it does not reach the needed quality, % (independent from the choice of the timeslot),
and $j_1$ remains incomplete in this case.

So let us choose the alternative assignment $j_0 \leftarrow \{i_1,i_2\}$ and set $i_2$ on $j_0$ at timeslot $0$:
\begin{center}
\begin{tabular}{rccc}
  timeslot & $0$ & $1$ & $2$  \\ \hline
  $i_0$ & ~ & ~ & $\times$ \\
  $i_1$ & ~ & $j_0$ & ~ \\
  $i_2$ & $j_0$ & ~ & $\times$ \\
\end{tabular}
\end{center}
Now $j_1$ cannot be completed without violating the sequential working assumption w.r.t.~this job.
On the other hand, if we set $i_2$ on $j_0$  at timeslot $2$, we can complete both jobs with the schedule 
\begin{center}
\begin{tabular}{rccc}
  timeslot & $0$ & $1$ & $2$  \\ \hline
  $i_0$ & ~ & ~ & $j_1$ \\
  $i_1$ & ~ & $j_0$ & ~ \\
  $i_2$ & $j_1$ & ~ & $j_0$ \\
\end{tabular}
\end{center}
without violating any constraints.
To complete the discussion, note that if we choose  $j_1 \leftarrow \{i_1,i_2\}$ in the beginning, then $j_0$ cannot be completed no matter what timeslot is used for $i_2$
 \textbf{(end of example)}.

The example shows that not only the choice of an optimal worker-task assignment without consideration of time may be misleading, but also that the specific selection of timeslots is important.

\subsection{The \kic Problem Model}
With the following definition we want to capture the interplay between task assignment \emph{and} timeline sequencing within the same model, and add appropriate constraints. We refer to our problem description as {\sc task assignment and sequencing} (\kic) in expert crowdsourcing.
\subsubsection{Input Data}
\textbf{Scheduling period}.
Suppose we look at $t$ timeslots $[t]=\{0,1,\ldots, t-1\}$. Each timeslot $d\in [t]$ is also called a \emph{day} but can be any fixed period of time.

\textbf{Knowledge domains}. A finite set $K$ of know\-ledge domains. Each $k\in K$ represents an area of know\-ledge or a knowledge topic. 

\textbf{Workers}. A finite set $U$ of users, hereby refered to as workers, participating in the crowdsourcing platform. Each worker $i \in U$ has the following characteristics\footnote{Note that in the context of this work the quantification of worker expertise, wage or speed are considered orthogonal to the studied assignment and sequencing problem. The interested reader is referred to~\cite{Lykourentzou201018,Dalip:2011:AAD:2063504.2063507,Ipeirotis:2014:QTC:2566486.2567988} for available worker profile quantification techniques based on machine-learning, implicit evaluation or information theory.}:
\begin{itemize}
\item \textit {Expertise}. An expertise vector $e_i$ of dimension $|K|$. The expertise $e_{ik}$ of a worker denotes the added quality that the worker can bring to a job belonging to domain $k$. 
\item \textit{Wage}. A cost  vector $w_i$ of dimension $|K|$.  The amount $w_{ik}$ is the monetary renumeration that the worker demands in order to perform a job  belonging to domain $k$. 
\item \textit {Availability}. An availability vector $a_i$ of dimension $t$ with entries $a_{id}=1$ if worker $i$ is available on day $d$,  and  $a_{id}=0$ otherwise.
%\item \textit{Speed}.  A vector $s_i$ of dimension $|K|$.  The working speed $s_{ik}$ is the time that the worker takes to complete a job belonging to domain $k$. %, once he has agreed to undertake it. 
%\item (NOTE: drop this parameter?) \textit{Maximium number of jobs}. A worker can only be assigned to $m_i$ many jobs in total during the scheduling period.
\end{itemize}

\textbf{Jobs}. A finite set $J$ of knowledge-intensive jobs that are crowdsourced. A job $j$ is assumed to have the following characteristics:
\begin{itemize}
 \item \textit {Domain}. Each job belongs to exactly one domain $k_j\in K$. 
\item  \textit {Quality threshold}.  The amount $Q_j$ is the minimum quality that the job needs to achieve. 
\item \textit {Cost threshold}. The budget for job $j$ is given by $C_j$ as the maximum total amount of money that can be paid for the job. 
		% The amount $C_j$ represents the initally available budget.
		%It will be punished by the objective function if more than $C_j$ is used to reach the quality threshold.
%\item \textit{Deadline}. With $D_j$ the maximum time is denoted until which the job needs to be finished. 
 \item \textit {Release date}. Each  job has a release date $r_j\in [t]$ which means that
	 job $j$ enters the crowd system at timeslot $r_j$ (and never leaves the system).
\end{itemize}
\textbf{Sequentiality}. Finally our model assumes a \emph{sequential work mode} along the timeline, according to which workers build on one another's contributions, at most one worker can be assigned to a task simultaneously, and each worker contributes at most to a single task at a time. 
Sequentiality is chosen for three reasons. First it is often imposed by the nature of expert crowdsourcing macrotasks, which are not easily decomposable to microtask level and as such they do not allow multiple simultaneous worker contributions (e.g. writing a document cannot be done by decomposing it to sentence level). 
Second, sequentiality allows building on the task's quality while not necessitating worker concurrency, which in practice is more difficult to achieve when specific worker skills (i.e. experts on a topic) are required. Third, sequentiality allows a more realistic coupling of our approach with worker skill evaluation mechanisms, making it easier to accurately evaluate the quality that each worker has brought once she has finished working on a task. Nevertheless, as also discussed in section~\ref{extensions} an extension of our model to include worker concurrency is feasible and we aim to examine it as part of our future work.

%%%%%%%%%%%%%%%%%%%%%%%%%%%%%%%%%%%%%%%%%%%%%%%%%%%%%%%%%%%%%%%%%%%%%%%%
\subsubsection{Feasible Solutions, Constraints and Optimization Goal}
A schedule needs to carry information about the resource allocation for each job in terms of workers and in terms of time: When does what worker contribute to which job?
% Due to various constraints both aspects are not independent of each other.

\textbf{Solutions}. In a solution for input data $x = (t,K,U,J)$ we have for each job $j\in J$ a vector $U_j$ of dimension $t$ with entries from $U\cup\{\nn\}$.
If $U_{jd}=i$ then worker $i$ is assigned to job j and scheduled on day $d$, and if $U_{jd}=\nn$ then there is no worker assignment for job $j$ on day $d$.

Note that we represent solutions hereby as job/timeslot-schedules with worker entries, whereas in the previous example we utilized an equivalent worker/timeslot-representation with
job entries.
So the successful schedule from the example in the present notation is
\begin{eqnarray*}
U_0 & = & (\nn, i_1, i_2)\\
U_1 & = & (i_2,\nn, i_0)
\end{eqnarray*}

\textbf{Constraints}. A solution is called \emph{feasible} if and only if the following holds:

\begin{enumerate} [label=(\alph*)]
\item \label{c1} No worker is assigned to more than one job at a time, i.e., for all $d\in [t]$ and distinct $j,j' \in J$ % with domains $k$ and $k'$, resp., 
	 it holds that $U_{jd} \neq U_{j'd}$ (unless both values are \nn).
	 % $t_{ij}+s_{ik} \leq  t_{ij'}$ or $t_{ij'}+s_{ik'}\leq  t_{ij}$.
\item \label{c2} No job is assigned to more than one worker at a time, i.e., for all $j\in J$ and $d\in [t]$ there is at most one worker stored in $U_{jd}$.
	This is ensured by the representation of $U_{j}$. 
\item \label{c3} No worker is assigned more than once to the same job, i.e., for all $j\in J$ and distinct $d,d' \in [t]$ % with domains $k$ and $k'$,  resp., 
	 it holds that $U_{jd} \neq U_{jd'}$ (unless both values are \nn).
\item \label{c4} No worker is scheduled on a day where she is not available, i.e., for all $d\in [t]$ and $j \in J$ it holds that
		if $U_{jd}=i$ then $a_{id}=1$.
\item \label{c5} No job is worked on before its release date, i.e., for all $j \in J$ and $d < r_j$ it holds that $U_{jd}=\nn$.
\item \label{c6} No job exceeds its budget, i.e., for all $j \in J$ it holds that $c_j \leq C_j$ where $c_j$ is the \textit{cost of job $j$} defined as $c_j=\sum_{i\in U_j} w_{ik}$
		if $j$ has domain $k$.

%\item \textcolor{red}{Can we also add non-preemption as a constraint here? This basically means that once a worker has been assigned a task, the worker cannot be pulled out of the task until the time thas she has to finish the task (her allocated time slot) ends.}
\end{enumerate}
Note that there always exists a trivial feasible solution with $U_{jd}=\nn$ for all $j,d$.

\textbf{Objective}.
In order to assess the quality of a feasible solution $y = \{ j \mapsto U_j ~|~ j\in J \}$ we 
determine for each $j$ with domain $k$ the \textit{quality of job $j$} w.r.t.~this solution as
$q_j=\sum_{i\in U_j} e_{ik}$.
Note that in the context of this work we define task quality as the sum of expertises of the workers that participate in it, using the additive skill aggregation model that is often used for expert sequential macrotasks such as document editing~\cite{Anagnostopoulos:2012:OTF:2187836.2187950,Lykourentzou:2013:IWA:2584064.2584070}. Other aggregation functions, including minimum, maximum or 
product~\cite{Anagnostopoulos:2012:OTF:2187836.2187950} could also be used to compute a task's quality, however their full examination is out of the scope of this work.

Now we set the measure for input $x=(t,K,U,J)$ and solution $y$ to
$$
m(x,y) = |\{ j \in J ~|~ q_j \geq Q_j \}| 
$$
which we want to maximize.
Therefore we count the number of jobs that reach their quality threshold within budget and that can be scheduled in a feasible way w.r.t.~constraints \ref{c1} to \ref{c6}.
We call such jobs \emph{completed}.

%%%%%%%%%%%%%%%%%%%%%%%%%%%%%%%%%%%%%%%%%%%%%%%%%%%%%%%%%%%%%%%%%%%%%%%%
%\subsection{Problem Analysis}
\subsection{TAS: An allocation problem with two aspects}
The \kic optimization problem combines aspects of two well-studied problems of different nature, reflecting resource allocation of workers on one hand, and
allocation of timeslots on the other.

\subsubsection{Allocation of Workers: The Multiple Knapsack perspective}
If we look only at worker allocation in our model, we can understand each job $j$ of domain $k$ with budget $C_j$ as a knapsack of this size that we need to fill with 
worker's expertises $e_{ik}$. Since the worker availabilities restrict the number of times a single worker can be packed, we have a \emph{bounded} version of
the {\sc multiple knapsack} problem~\cite{Kellerer:2013wj}. The difference to this classical problem is the optimization goal. While in \kic we want to maximize the number of completed jobs with respect to their individual quality thresholds $Q_j$, the goal in {\sc multiple knapsack} is to maximize the sum of all packed expertises, no matter how these spread over the different knapsacks.

\subsubsection{Allocation of time slots: The Openshop perspective}
On the other hand, let us suppose a worker-task-assignment is already fixed such that all jobs reach their quality thresholds, and we need to schedule the selected workers along the timeline with respect to job releases and worker availabilities. Then we can understand this partial problem as a machine-scheduling problem: Here workers play the role of machines and jobs need to be processes on these machines. 
Observe that the order of processing is immaterial in our model, that we demand sequentiality, and that the processing time of a job on a certain machine is either $0$ or $1$ per timeslot (depending on whether the respecting worker is assigned to this job or not). So this aspect of \kic is a {\sc unittime openshop} problem with limited machine-availability and job release-dates \cite{Kravchenko:2000vu}. Note that the adoption of a model also implies non-preemption, i.e. a worker cannot be interrupted once he/she has started working on a task.
The goal of maximizing the number of completed jobs translates to minimizing the number of late jobs if we set $t$ as the gobal deadline.
We also want to mention that the sequencing of an already fixed worker-task-assignment can be reduced to the {\sc bipartite list edge-coloring} problem \cite{4567876}.
Here jobs and workers form a bipartite graph with the worker-task-assignments as it's edges, and timeslots are represented by colors. Then we
assign a list of colors to each edge $(j,i)$ such that worker $i$ is available on these timeslots and job $j$ is already released. 
A proper coloring of all edges can be found if and only if the previously fixed worker-task-assignment can be sequenced on the timeline.

\subsubsection{TAS Complexity}
Both aspects of \kic that we have pointed out above are NP-hard on their own, so is \kic as we show below.
For an upper complexity bound note that the length of \kic-solutions is polynomially bounded in the input length and that the constraints can be checked in
polynomial time if a solution is given, so \kic is an NP-optimization problem.
Moreover, we observe that \kic is a large number problem, since it has {\sc knapsack} as a subproblem (if there is only a single job and each worker is available on a different single day).
So it is reasonable to consider strong NP-hardness.

\newcommand{\dm}{$\textrm{\sc 3-DM}$}
\newcommand{\dmlong}{$\textrm{\sc 3-Dimensional Matching}$}

\begin{theorem} \label{thm1}
\kic~is a strongly NP-hard optimization problem.
\end{theorem}

\begin{proof}
We show NP-hardness with a polynomial-time many-one reduction from the NP-complete 
problem \dmlong~\cite{Karp:1972tu}. 
For finite, disjoint sets $X$, $Y$ and $Z$ we say that $M\subseteq X\times Y\times Z$ is a 3-dimensional matching
if for all distinct triples $(x_1,y_1,z_1), (x_2,y_2,z_2) \in M$ it holds that $x_1\neq x_2$, $y_1\neq y_2$ and $z_1\neq z_2$.
It is known that \dmlong~is NP-complete even in the special case when •
$|X| = |Y| = |Z| = u$ and $M$ has to be a perfect matching with $|M|=u$.

%In the previous section we showed that our basic model \bss~is solvable in polynomial time. However, in practice there may be more constraints, which must be respected, therefore we analyzed the complexity of the extended versions \bssc, \bsst~and \bsse, too. 

%We begin the complexity analysis of the extended versions of \bss~by a polynomial time many-one reduction of the $\NP$-complete \dmlong~(\dm) to the \bssc. The \dm~is defined as follows:
%We begin the complexity analysis of the extended version \bssc~with a polynomial time many-one reduction of the $\NP$-complete \dmlong~(\dm) to \bssc. The \dm~is defined as follows:

\vspace{0.3cm}
\noindent\dmlong~(\dm)

\noindent
 \begin{tabular}{@{}lp{5.5cm}}
          \bf{Input:} & Finite and disjoint sets $X, Y, Z$ with $|X| = |Y| = |Z|$ and a subset 
          $W \subseteq X \times Y \times Z$.\\
          \bf{Question:} & Is there a perfect 3-dimensional matching $M\subseteq W$ ?
\end{tabular}

Suppose an instance of \dm~is given with $X=\{x_i ~|~ i\in[u]\}$, $Y=\{y_i ~|~ i\in[u]\}$, $Z=\{z_i ~|~ i\in[u]\}$
for some $u\geq 1$ and $W \subseteq X \times Y \times Z$. 
The idea is to use constraint \ref{c1} (no worker is assigned to more that one job at a time) to achieve the needed matching condition.
We take elements from $X$ ($Y$, $Z$) as workers available on day $0$ ($1$, $2$, resp.) and use domains to fix the given triples from $W$.
More precisely, we define a corresponding \kic-instance $(t,K,U,J)$ as follows:
\begin{itemize}
\item The scheduling period has $t=3$ timeslots.
\item There are $|W|$ many different domains in $K$.
\item Each triple $w\in W$ is encoded as a job $j_{w}$, and all jobs have pairwise different domains. 
		 For all jobs $j_{w}$ we set quality and cost threshold to $Q_{j_w}=C_{j_w}=3$ and release date to $r_{j_w}=0$.
\item Workers are defined as $U=X\cup Y\cup Z$. For $x\in X$, $y\in Y$ and $z\in Z$ we set the availability to $a_x=(1,0,0)$, $a_y=(0,1,0)$ and $a_z=(0,0,1)$, respectively.
		To fix expertise and wage, we consider each triple $w=(x,y,z)\in W$ and the corresponding job $j_w$. 
		If $j_w$ has domain $k$ then we define $e_{xk}=e_{yk}=e_{zk}=1$ and $w_{xk}=w_{yk}=w_{zk}=1$.
		All entries in expertise and wage vectors that are not addressed hereby are set to $0$.
\end{itemize}
First observe that a job $j_w$ with $w=(x,y,z)$ reaches it's quality threshold if and only if we assign workers $\{x,y,z\}$ to this job, since 
exactly these workers contribute to the job's domain.

Now we argue that the given \dm~instance has a perfect matching $M$ if and only if the constructed \kic instance has a feasible solution with $|M|=u$ completed jobs.
If $M\subseteq W$ is a 3-dimensional matching, then we consider the \kic-solution $U_{j_w}=(x,y,z)$ for all $w=(x,y,z)\in M$.
Since $M$ is a matching all distinct solution vectors differ in all components, so constraint \ref{c1} is satisfied. 
All other constraints are easy to check, just note that each worker is available only on a single day, that all jobs are immediately released and 
that no job can exceed the budget. All jobs in this solution are completed due to our previous observation.

Conversely, note that if there is a feasible \kic-solution with completed jobs $j_w$ and $w = (x,y,z)$, then 
it must be that $U_{j_w}=(x,y,z)$.
Since constraint \ref{c1} holds, the solution vectors for any two distinct jobs differ in all components.
So $M=\{(x,y,z)~|~U_{j_w}=(x,y,z) \text{~and $j_w$ completed} \}$
is a 3-dimensional matching and $|M|=u$.

The reduction function maps only to \kic-instances where all integer values are polynomially bounded in the input length, so strong NP-hardness follows.
\end{proof}

This rules out the possibility of pseudo-polynomial algorithms and the existence of fully polynomial-time approximation schemes for \kic unless P equals NP.
Furthermore note that the reduction emphasizes the aspect of timeline sequencing, since worker-task-assignments in the constructed \kic-instance are trivial
(there is exactly one feasible worker-assignment possible to reach the quality threshold of each job).

\section{An Online Algorithm for \kic}
\label{algorithm}

Due to the dynamic nature of crowdsourcing systems, it seems not realistic to consider \kic as an \emph{offline} problem where algorithms are provided with the complete input at once. 
In fact, worker availabilities are hardly predictable and it is usually not known in advance which jobs will enter the system at what time.
So the problem of task assignment and sequencing is inherently \emph{online} in nature and sequencing decisions have to be taken without complete information about the input data.
We say that an algorithm for \kic has the \emph{online property}, if it processes the input in a serial way w.r.t. the timeline $d = 0,1,\ldots$ and 
in each step $d$ the algorithm has to take its assignment decisions while having access only to the time-dependent information of the input for timeslots $\leq d$.
These are the worker availabilities and the jobs released up to day $d$.
For more background on the general concept of online algorithms we refer to \cite{Borodin:2005wx}.

To design such an algorithm we start with the following observation:
Suppose a feasible solution $y$ for \kic is given. If we look at a single day $d$ in this solution then the assignments of workers to jobs
for that day form a bipartite matching between the (uncompleted) jobs with (remaining) quality needs and budget
	on one hand, and the set of available workers for that day on the other hand. Constraints \ref{c1} and \ref{c2} form exactly this bipartite matching condition.
	
So conversely, if we proceed day by day with our online algorithm, we can try to compute a matching between the active (= released but incomplete) 
jobs $J'$ in the system on that day, and the available workers $U'$ for that day. 
Note that due to this choice of $J'$ and $U'$ we also immediately satisfy constraints~\ref{c4} and \ref{c5}.
It remains to consider constraints \ref{c3} (no worker assignment to the same job twice) and \ref{c6} (no job exceeds it's budget).
Both can be taken care of when we construct the edges of possible assignments in the bipartite graph between $J'$ and $U'$: If the remaining budget for a job is smaller than
the wage of a worker in this domain, then the edge is omitted. 
The same is true if the worker has already been assigned to this job in the past.
Both conditions can be checked when looking at the partial solution for timeslots $< d$.
Together, this online procedure results in a series of matchings $M_d$ for $d=0,1,\ldots,t-1$ that form a feasible solution $y$ for \kic.

More than that, we want to choose a sequence of matchings that yields a large number of completed jobs.
Among all possible matchings for each day $d$, which is the right one? 
We propose a greedy approach and compute in each step a matching, such that the sum of \emph{profits} we get from the respective assignments 
for that day is maximized. 
More precisely, we construct for each day a \emph{weighted} bipartite graph where each possible assignment (edge) claims a certain profit.
In our algorithm, the profit is just the amount of expertise per wage unit (efficiency).
The problem {\sc max weighted bipartite matching} can be solved to optimality by known algorithms in polynomial time, e.g. if we apply the \emph{Hungarian Method} this step has a running time
proportional to $O((|J'|+|U'|)^2|E|)$ \cite{Kuhn:1955to}.
So we obtain the following online Algorithm~\ref{alg:tas-online} for \kic with polynomial running time $O(t|J|^3|U|^3)$.

\begin{algorithm}[ht]
 \caption{{\sc tas-online}}
  \label{alg:tas-online}
  \begin{algorithmic}[1]
    \Require{A \kic-instance $x=(t,K,U,J)$}
    \Ensure{A feasible solution $y$ for $x$.}
    \Statex
    \State \emph{Set $U_{jd} = \nn$ for all $j\in J$ and $d\in[t]$.}
     \For{$d=0,1,\ldots,t-1$}
 	  \Comment{proceed day-by-day}
	 \State \emph{$J' = $ uncompleted jobs with  $r_j\leq d$} 
	 
	 		\Comment{active jobs}
	 \State \emph{$U' = $ workers available on day $d$}   
	 
	  		\Comment{active workers}
	  \State $E= \emptyset $ \Comment{edge set in bip. graph}
	   \For{$(j,i)\in J' \times U'$}
		\If{$i \in U_j$} \algorithmicpass \Comment{ensures  \ref{c3}} \label{online:feasible1}
		\EndIf
	   	\If{$e_{ik_j}==0$} \algorithmicpass 
		
		\Comment{$i$ has no expertise in $j$'s domain}
		\EndIf
		\If{$w_{ik_j} > C_j - c_j$} \algorithmicpass \label{online:feasible2}
				\Comment{ensures  \ref{c6}}
		\EndIf
		\State $\mathit{profit} \gets e_{ik_j} / w_{ik_j}$
	   	 \State $E\gets (j,i, \mathit{profit})$
	   \EndFor
	    \State $M_d \gets $ \emph{MaxWeightedMatching($J',U',E$)}
	     \For{$(j,i)\in M_d$}
		     \State $U_{jd} = i$ 	\Comment{worker-task-assignment}
	   \EndFor
    \EndFor
    \State\Return{$\{ j \mapsto U_j ~|~ j\in J \}$}
  \end{algorithmic}
\end{algorithm}

This algorithm can be viewed as an online schema that allows multiple extension, which we discuss in the last section after some experimental evaluation using the 
present basic version.

%%%%%%%%%%%%%%%%%%%%%%%%%%%%%%%%%%%%%%%%%%%%%%%%%%%%%%%%%%%%%%%%%%%%%%%%
%%%%%%%%%%%%%%%%%%%%%%%%%%%%%%%%%%%%%%%%%%%%%%%%%%%%%%%%%%%%%%%%%%%%%%%%

\section{Experimental Evaluation}
\label{results}
With the hardness result we have already seen that \kic has the {\sc KNAPSACK} decision problem as a subproblem. It is known from literature that no competitive algorithm for the \emph{online} version of  {\sc KNAPSACK} exists where items arrive one at a time  \cite{MarchettiSpaccamela:1995ih}.
An online algorithm is called competive if the ratio of it's performance and an optimal offline algorithm's performance can be bounded, a usual performance measure for online algorithms \cite{Borodin:2005wx}. 
It follows that no competitive online algorithm for \kic exists as well. Therefore, in order to evaluate {\sc tas-online} experimentally, we formulate alternative algorithms to compare with.
\subsection{Experimental Setup}
\subsubsection{Benchmarks}
We evaluate the performance of {\sc tas-online} using four benchmarks, with each benchmark extending the previous with a new functionality. 
The first version of the  algorithm ({\sc random}) 
builds a feasible solution randomly and without any individual preferences of workers that usually appear in a fully self-organized system.
We simply iterate over the available workers and pick a feasible job.

\begin{algorithm}[ht]
 \caption{{\sc random}}
\label{alg:random}
  \begin{algorithmic}[1]
    \Require{A \kic-instance $x=(t,K,U,J)$}
    \Ensure{A feasible solution $y$ for $x$.}
    \Statex
    \State \emph{Set $U_{jd} = \nn$ for all $j\in J$ and $d\in[t]$.}
     \For{$d=0,1,\ldots,t-1$}
 	  \Comment{proceed day-by-day}
	 \State \emph{$J' = $ uncompleted jobs with  $r_j\leq d$} 
	 
	  		\Comment{active jobs}
	 \State \emph{$U' = $ workers available on day $d$}  
	 
	 		\Comment{active workers}
	   \While{$U'\neq \emptyset$}
	   	 \State \emph{pick worker $i\in U'$ randomly, remove it}
		 \State \emph{$J'_i$ = feasible jobs for worker $i$} \label{random:feasible}
		 \State \emph{pick job $j\in J'_i$ randomly}  \label{random:pick}
		 \State $U_{jd} = i$ 	\Comment{worker-task-assignment}
		 \If{$q_j\geq Q_j$} \emph{remove $j$ from $J'$} 
		 
		 				\Comment{remove completed jobs}
		\EndIf

	   \EndWhile
    \EndFor
    \State\Return{$\{ j \mapsto U_j ~|~ j\in J \}$}
  \end{algorithmic}
\end{algorithm}

To obtain the feasible jobs for $i$ in line~\ref{random:feasible} we proceed as in lines~\ref{online:feasible1} to \ref{online:feasible2} in {\sc tas-online}
and additionally check that $j$ is still without worker assignment for that day.

For the next version of the algorithm ({\sc random egoistic}) we assume that workers try to realize a larger wage with priority, thus modeling a typical crowdsourcing environment,
where workers are self-appointed to tasks, trying to maximize their individual profit~\cite{DBLP:conf/icwsm/RogstadiusKKSLV11}. To do so, we substitute  line~\ref{random:pick} in the previous algorithm 
with the lines stated in Algorithm~\ref{alg:randego}.

\begin{algorithm}[ht]
 \caption{{\sc random egoistic}}
 \label{alg:randego}
  \begin{algorithmic}[1]
  \makeatletter
\setcounter{ALG@line}{79}
\makeatother
		 \State \emph{let $k_0, k_1, \ldots $ be the domains sorted decr.~by $i$'s wage}
		    \For{$k=k_0,k_1,\ldots$}
		    	\State \emph{$J'_i$ = feasible jobs for worker $i$ from domain $k$}
			 \If{$J'_i \neq \emptyset$} 
			 	\State \emph{pick job $j\in J'_i$ randomly}
				\State \algorithmicbreak		 
			\EndIf
		    \EndFor
  \end{algorithmic}
\end{algorithm}

In the next step ({\sc random egoistic filter}), we extend {\sc random egoistic} with a filter that restricts the jobs that are offered to each worker based on expertise. This models the practice employed by many crowdsourcing platforms today, where workers can only access a task if they successfully pass a ``screening" (realized through the use of performance levels, golden data, reputation, or other means across the different platforms)~\cite{Downs:2010:YPG:1753326.1753688,Josang:2007:STR:1225318.1225716}, which allows to expect a substantial contribution to the job's quality by these workers.
This ``screening threshold" is expressed by an additional parameter $0<\mathit{factor}<1$ which determines the minimal expertise needed. Therefore we additionally substitute  line~\ref{random:feasible} 
with the following lines.

\begin{algorithm}[ht]
 \caption{{\sc random egoistic filter}}
%  \label{alg:zuweisung}
  \begin{algorithmic}[1]
  \makeatletter
\setcounter{ALG@line}{69}
\makeatother
		 \State \emph{$J'_i$ = feasible jobs for worker $i$}
		 \State \emph{remove all $j$ from $J'_i$ with $e_{ik} < (Q_j \cdot  \mathit{factor})$} 
  \end{algorithmic}
\end{algorithm}

As as last variation of Algorithm~\ref{alg:random} we choose a job for some worker completely deterministically with a greedy rule: 
Job $j$ is assigned to worker $i$ if the marginal contribution in terms of quality is maximal among all feasible jobs for worker $i$.
This amounts to replacing line~\ref{random:pick} in Algorithm~\ref{alg:random} by the following line (an leave line~\ref{random:feasible} unchanged).

\begin{algorithm}[ht]
 \caption{{\sc online greedy}}
%  \label{alg:zuweisung}
  \begin{algorithmic}[1]
  \makeatletter
\setcounter{ALG@line}{79}
\makeatother
		\State \emph{pick job $j\in J'_i$ such that $e_{ik_j} - q_j$ is maximal}
  \end{algorithmic}
\end{algorithm}

%\textcolor{red}{TO ADD: To add the Online-Greedy algorithm, which assigns (to check again) tasks to workers taking into account the worker's marginal utility on each task, i.e. how much the worker's assignment improves the objective function of the task (potentially extended for comparison reasons to also greedily take into account the task's slack time)}  

Note that all algorithms so far have the online property for \kic.
Finally, for reasons of comparison, we use an \emph{offline} algorithm (Algorithm~\ref{alg:tas-offline}) that does not have to proceed day-by-day but has access to the complete input at once.
So this clairvoyant algorithm knows in advance what workers will be available on what days of the scheduling period, and also what jobs will eventually enter the system.
It proceeds job-by-job and treats each job as a knapsack that has to be packed with workers (items) that are available after the job's release date. 
To obtain such a packing, it calls an optimal algorithm for {\sc max knapsack} that returns a packing with minimal cost such that the quality threshold is reached.
Then, for the workers from this packing (= worker-task assignment), a sequencing on the timeline w.r.t. their availability is fixed, before the next job is considered.

The algorithm has two more parameters that influence the way workers are selected for input to the knapsack algorithm for a job $j$. With $\mathit{lookahead}$ we specify the
interval of timeslots $[r_j, r_j\!+\!\mathit{lookahead}]$  from which the available workers are chosen in order to control the maximum flow time of each job.
Secondly, we use $\mathit{minavail}$ to ensure that each worker has at least $\mathit{minavail}$-many free timeslots remaining in the above interval in order to facilitate
the allocation of timeslots afterwards.

\begin{algorithm}[ht]
 \caption{{\sc tas-offline}}
\label{alg:tas-offline}
  \begin{algorithmic}[1]
    \Require{A \kic-instance $x=(t,K,U,J)$}
    \Ensure{A feasible solution $y$ for $x$.}
    \Statex
    \State \emph{Set $U_{jd} = \nn$ for all $j\in J$ and $d\in[t]$.}
     \For{$j \in J$}
 	  \Comment{ordered by release dates}
	 \State \emph{$U_j' = $ feasible workers for $j$}
	  \For{$i \in U_j'$}
	  	 \If{\algorithmicnot $\mathit{minavail} \,\text{in}\, [r_j, r_j\!+\!\mathit{lookahead}]$} 
		  \State  \emph{remove $i$ from $U_j'$}
	  	\EndIf
	  \EndFor
	  \State $W_j \gets $ \emph{DynProgKnapsack($U'_j,Q_j, C_j$)}
	   \For{$i\in W_j$} 			\Comment{sorted decr.~by expertise}
	   	     \State  \emph{$d = $ earliest available timeslot $\geq r_j$ for $i$}
		     \State $U_{jd} = i$ 	\Comment{worker-task-assignment}
		     \State $a_{id} = 0$ 	\Comment{set $i$ unavailable on $d$}
	   \EndFor
    \EndFor
    \State\Return{$\{ j \mapsto U_j ~|~ j\in J \}$}
  \end{algorithmic}
\end{algorithm}

The offline algorithm does not guarantee optimal solutions for \kic for various reasons. 
However, it is designed to complete as many jobs as possible by the particular use of off\-line information and by incorporating optimal solutions to the knapsack subproblems.
Note that due to the standard dynamic-programming (DP) algorithm for   {\sc max knapsack} this is only a 
 pseudo-polynomial time algorithm (the DP-table has dimension $|U|\times C_j$ for each job) \cite{Kellerer:2013wj}.
While we observe that this still yields tolerable runtimes for realistic input sizes, it is also possible to scale down the range of cost thresholds, or to use
a fully polynomial-time approximation-scheme instead, if runtime becomes crucial.

\subsubsection{Evaluation Metrics and Experiments Overview}
To evaluate our approach we compare the algorithms principally in terms of the objective function value (i.e. the metric that the \kic model is meant to optimize), both as an absolute number and as a percentage of the upper bound of completable jobs. We also use four auxiliary metrics, meant to provide more information on the algorithm's behavior: the number of assigned workers, flow time, budget utilization and quality reached. 

We conduct two types of experiments: i) synthetic (sections~\ref{synthetic} and~\ref{scalability}), where we experiment with a known simulated crowdsourcing instance and its variations and ii) real-world (section~\ref{realexp}), where we examine our model on an actual crowdsourcing platform. 

\subsection{Synthetic Data Experiment}
\label{synthetic}
\subsubsection{Simulation parametrization}
\label{parametrization}
We first experiment with synthetic data, which were generated using the experimental result distributions reported in~\cite{roy2015}, where AMT workers worked on the complex task of news writing. For simplicity, all modeling elements were generated in the $[0.1]$ scale. Worker expertise received a random value from a normal distribution with mean equal to $0.5$ and a variance $0.15$, while worker wage received a random value from a normal distribution with mean equal to $0.5$ and variance $0.2$. For this set of experiments worker acceptance was set equal to $1$. Job quality threshold was modeled using a beta distribution with $\alpha=5$, $\beta=1$, so that most jobs require a quality of at least $0.6$ of $1$ and higher with only a tail of jobs requiring less. Job cost threshold was then modeled as linearly related to job quality. Worker and job arrivals were modeled as Poisson processes with an average $\lambda=200$ worker/day, and $\mu=20$ jobs/day, respectively. Overall, we simulated a timeline of $30$ days, during which $1000$ workers (re)entered the system and $600$ jobs were requested, belonging to $10$ knowledge domains.
So we have the following numbers in terms of our model:

\begin{center}
\begin{tabular}{|c|c|c|c|}
 $t$ & $|K|$ & $|U|$ & $|J|$  \\ \hline
  $30$ & $10$ & $1000$ & $600$ \\
\end{tabular}
\end{center}

\subsubsection{Upper bound calculation}
\label{upper_bound}
First we compute an \emph{upper bound} on the number of ultimately completable jobs using the optimal DP-algorithm for {\sc max knapsack}:
%We can use the optimal DP-algorithm for {\sc max knapsack} to compute an \emph{upper bound} on the number of completable jobs as follows:
Assume for each job that this job is the first for which we compute a worker-task assignment, i.e., all workers with at least one available timeslot $\geq r_j$ are possible
knapsack items regardless of any other assignments. Now if the DP-algorithm does not find a packing within budget and above the quality threshold
with this input data, then this job cannot be completed whatsoever.
For the present instance, it turns out that at most $515$ out of the $600$ jobs can be completed.

%\textcolor{red}{Lower bound? Does it make sense to try to compute it?}

\subsubsection{Quality experiments}
\label{quality_experiments}
We now conduct the quality experiments. In terms of the objective function, we observe that {\sc tas-online} does not reach the number of completed jobs of our offline algorithm, but that it is significantly better that the other online algorithms 
(cf. Table~\ref{table_obj}).
%We see that {\sc tas-online} does not reach the number of completed jobs of our offline algorithm, but that it is significantly better that the other online algorithms 
%(cf. Table~\ref{table_obj}).
%

\begin{table}[ht]
\caption{Objective function (completed jobs)}
\begin{center}
\begin{tabular}{l|c|c}
Algorithm & absolute & $\%$ of bound \\ \hline
{\sc random} & $2$ & $0,39$ \\
{\sc random egoistic} & $98$ & $19,03$ \\
{\sc random ego. filter} & $114$ & $22,14$ \\
{\sc online greedy} & $82$ & $15,92$ \\
{\sc tas-online} & $\mathbf{355}$ & $\mathbf{68,93}$ \\
{\sc tas-offline} & $411$ & $79,81$ \\
\end{tabular}
\end{center}
\label{table_obj}
\end{table}%

To get a more precise picture, we want to compare these algorithms not only w.r.t. this single measure, but also look at other characteristics.
Next we ask how many workers are assigned to each job, and how long the flow times are, i.e., the number of timeslots
between release date $r_j$ and the latest assigned worker for $j$ (cf. Table~\ref{table_workers}).
In both cases we take the average values over all jobs in the system (not only the completed ones).

\begin{table}[ht]
\caption{Number of assigned workers and flow time}
\begin{center}
\begin{tabular}{l|c|c}
Algorithm & workers  & flow time \\ \hline
{\sc random} & $4,77$ & $4,77$ \\
{\sc random egoistic} & $3,46$ & $3,46$ \\
{\sc random ego.~filter} & $1,64$ & $7,37$ \\
{\sc online greedy} & $3,56$ & $2,91$ \\
{\sc tas-online} & $\mathbf{3,31}$ & $\mathbf{3,31}$ \\
{\sc tas-offline} & $2,41$ & $8,11$ \\
\end{tabular}
\end{center}
\label{table_workers}
\end{table}%

In cases where both values are the same, we only have compact assignments per job without any free slots in between.
While {\sc tas-online} seeks this type of assignments  we note that {\sc tas-offline} creates notable 
slack times, presumably a price to pay for larger number of completed jobs.

Now we state how much budget is used with these assignments, and how much quality is reached,
both relativ to the given thresholds and on average over all jobs in the system (cf. Table~\ref{table_budget}).

\begin{table}[ht]
\caption{Budget usage and reached quality in $\%$}
\begin{center}
\begin{tabular}{l|c|c}
Algorithm & budget  & quality   \\ \hline
{\sc random} & $88,16$ & $60,62$ \\
{\sc random egoistic} & $92,05$ & $90,27$ \\
{\sc random ego. filter} & $52,6$ & $55,77$ \\
{\sc online greedy} & $91,44$ & $87,12$ \\
{\sc tas-online} & $\mathbf{94,35}$ & $\mathbf{97,76}$ \\
{\sc tas-offline} & $67,73$ & $70,21$ \\
\end{tabular}
\end{center}
\label{table_budget}
\end{table}%

Due to its greedy nature {\sc tas-online} reaches very high quality values including for incompleted jobs and exploits the given budgets to a large extend.

Finally, we show how the main performance measures develop over the time, see Fig.~\ref{fig:completion_rate} for completed jobs and  Fig.~\ref{fig:quality_rate}  for reached quality.
Interestingly, we observe that the higher values in Fig.~\ref{fig:completion_rate} for {\sc tas-offline} appear towards the end of the scheduling period. 
A possible explanation is that the lookahead mechanism of this algorithms takes the end of the timeline into account.

\begin{figure}[h]
\center
\includegraphics[scale=0.48]{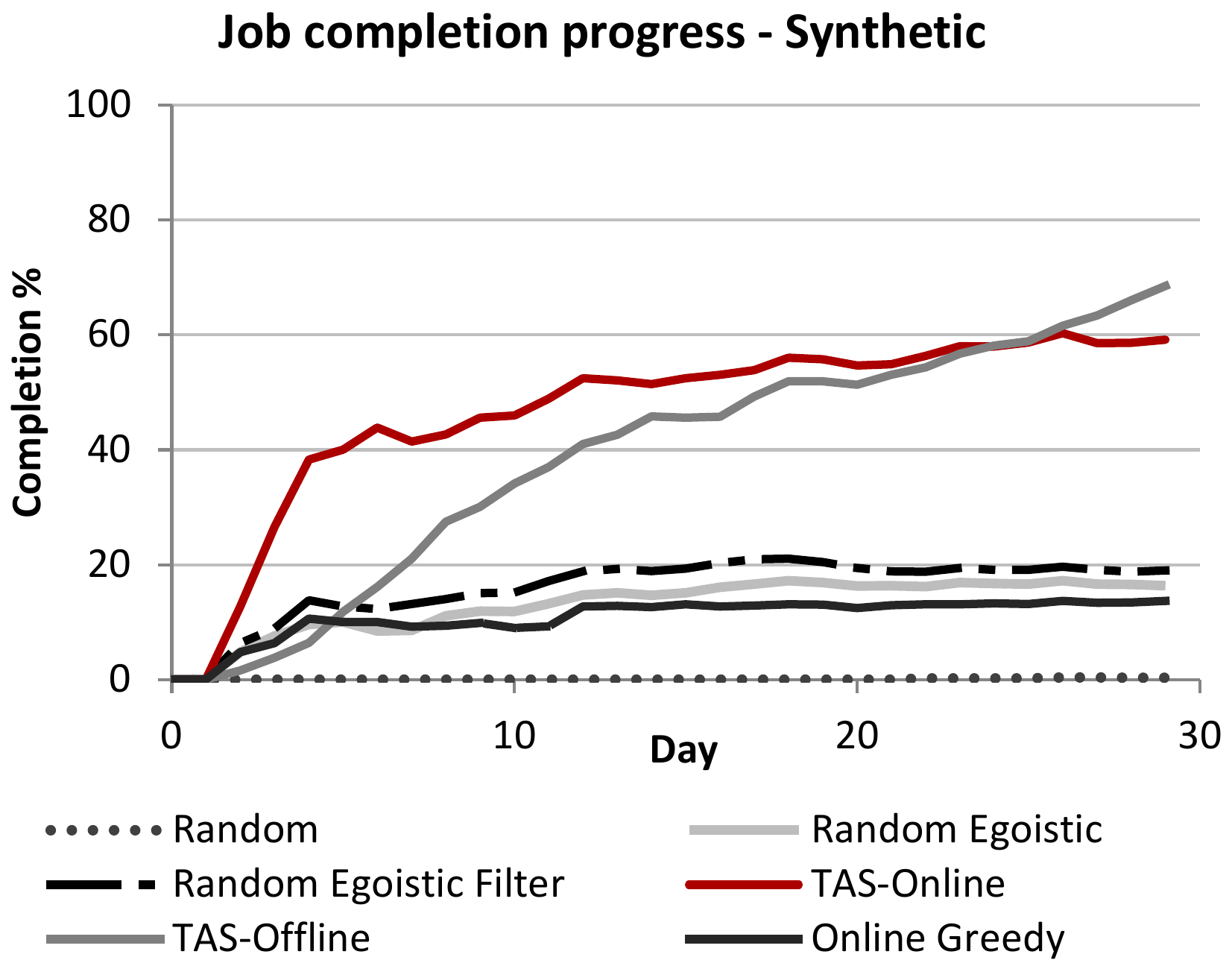}
\caption{Job completion over time for each of the five tested algorithms. The proposed algorithm {\sc tas-online} manages to achieve more completed jobs most of the time compared to its competitors. 
Time unit expressed in days. }
\label{fig:completion_rate}
\end{figure}

\begin{figure}[h]
\center
\includegraphics[scale=0.48]{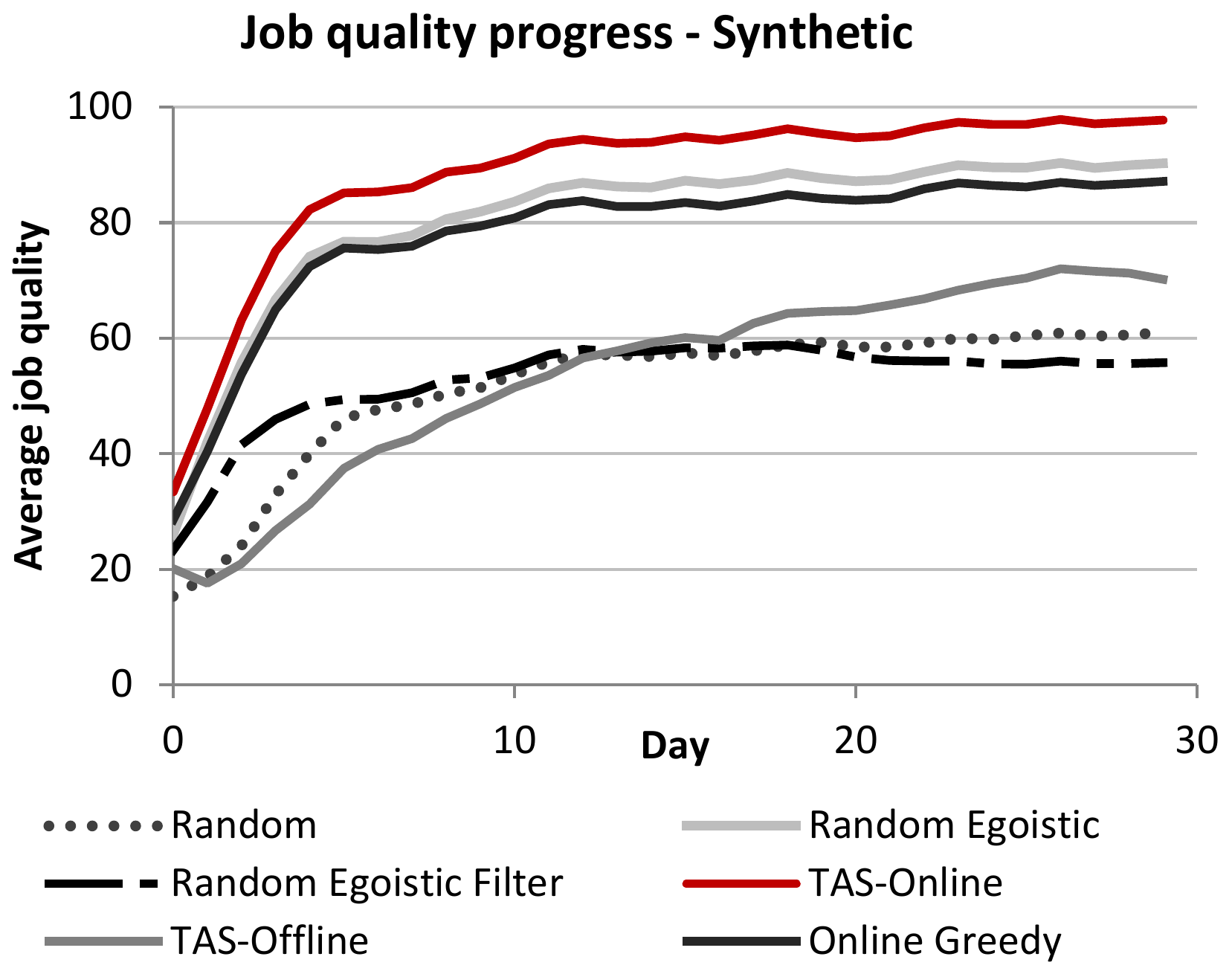}
\caption{Average job quality over time for each of the tested algorithms. The proposed algorithm {\sc tas-online}  manages to achieve higher quality per time unit compared to its competitors. Time unit expressed in days.}
\label{fig:quality_rate}
\end{figure}

\subsection{Scalability experiments}
\label{scalability}
Next we perform a series of scalability experiments to examine the robustness of our proposed algorithm under varying conditions of the simulated instance. Given that worker volatility is the most uncontrollable factor in crowdsourcing, the two parameters that we vary are: the available expertise and the available number of workers. Each parameter is modified independently, while all the other parameters of the baseline instance presented in 
section~\ref{synthetic} are kept the same. 
%sections~\ref{parametrization}-\ref{quality_experiments} are kept the same. 
%Modifying the available expertise means that we vary the average expertise levels of the crowd worker population, while keeping their numbers equal to the baseline instance. Modifying the available number of workers means that we vary the quantity of workers that are online while keeping their average expertise equal to that of the baseline instance. 
The variables that we measure are also the same as those measured for the baseline instance and include the objective function, as well as budget utilization, total flow time, number of assigned workers per task and percentage of the average quality threshold reached. 

The scalability experimental results are illustrated in Figures~\ref{fig:scal_obj_function1}-\ref{fig:qual_vs_exp}. For each of those figures the $x$ axis corresponds to the varied parameter, the $y$ axis to the measured variable and the vertical line at $x=1$ corresponds to the results of the baseline instance reported in section~\ref{quality_experiments}.

\subsubsection{Overview of scalability experimental results}
Two main remarks can be drawn as an overview of the scalability experiments. The first is about \textit{performance}: {\sc tas-online} is the highest performing among its online competitors, both regarding the value of the objective function, i.e.~the metric that the algorithm is meant to optimize, and on quality, without significant compromises on any of the remaining metrics. {\sc tas-online} is the only algorithm among those examine to achieve this: whereas certain algorithms come close to its performance for certain metrics and parameter values, the same algorithms are significantly low-performing in other metrics and parametrizations. This result indicates that the proposed algorithm has a better `value-for-money' compared to its competitors. 

The second remark is about \textit{consistency}: {\sc tas-online} is not only more performant, but its performance is consistent across the varying values of the scalability experimental parameters. This result indicates that the performance of the algorithm as detailed in section~\ref{quality_experiments} is not incidental but an inherent property of the algorithm, and reinforces trust in the algorithm's future usage. In the following we present a detailed analysis of the scalability experiments.

\subsubsection{Scalability effect on Objective Function: {\sc tas-online} gets consistently more jobs done}
In Figure~\ref{fig:scal_obj_function1} we measure the value of the objective function (i.e.~the number of accomplished jobs) as we modify expertise availability, and higher $y$ axis values are better. As we can observe, the {\sc tas} family of algorithms (both the online and the offline version) are able to achieve and maintain higher performance than their competitive algorithms, at all expertise levels. For average expertise levels less than the baseline instance {\sc tas-offline} is the best-performing algorithm, followed closely by {\sc tas-online}, while for expertise levels slightly higher than the baseline {\sc tas-online} takes and maintains precedence. As it can be expected, as the average worker expertise per knowledge domain drops, the performance of all algorithms drops steeply as well. Nevertheless, we can also observe that the {\sc tas} algorithms are more robust, in the sense that they maintain their high performance when the other algorithms already start losing theirs (notice for example the almost unchanged performance of {\sc tas-online} between $x=2$ down to $x=1.2$ compared to the steep performance drop of the other algorithms in the same range). 

A similar pattern can be observed when modifying the worker availability parameter (Figure~\ref{fig:scal_obj_function2}). In this case too, {\sc tas-online} is by far the most performant of all the online algorithms, surpassed only by its offline version. In fact, the performance difference between {\sc tas-online} and the rest of the algorithms is quite striking here, as {\sc tas-online} reaches approximately 60\% of the objective function value while the rest of the algorithms only reach 20\%. A second interesting remark that can be derived is that worker availability seems to have little effect on the algorithms after a certain critical mass of crowd workers has been gathered (which for our simulation corresponds to $x=0.4$, i.e. 40\% of the population of the baseline instance). These two observations (superiority of the {\sc tas-online} and small effect of worker availability after a certain critical mass) also hold when we measure the effect of the worker availability parameter on all other variables of the scalability experiment. Following this, and for reasons of brevity, we omit the rest of the scalability figures corresponding to the worker availability, and focus on the parameter of expertise availability which seems to have the highest effect.

\subsubsection{Scalability effect on Quality: {\sc tas-online} achieves higher quality}
Figure~\ref{fig:qual_vs_exp} illustrates the average task quality (expressed as the percentage of the quality threshold reached) for every level of expertise of the crowdsourcing population, and higher $y$ axis values are better. We observe that {\sc tas-online} manages to achieve the highest quality levels, surpassing even its offline version, for all expertise levels. In fact, given a certain level of expertise ($x=1.2$) and above, the algorithm manages to surpass the quality threshold set for the tasks. {\sc random-egoistic} and {\sc online greedy} are the second and third most performing algorithms respectively, but unlike {\sc tas-online} they achieve their high quality results, at the cost of accomplishing too few jobs, as it can be seen by juxtaposing Figures~\ref{fig:scal_obj_function1} and \ref{fig:qual_vs_exp}.

\subsubsection{Scalability on other parameters: {\sc tas-online} performs comparably to its competitors}
\textbf{Effect on cost}. We now examine the effect that the modification of expertise availability has on the budget used by the allocation algorithms (Figure~\ref{fig:budg_vs_exp}, smaller $y$ axis values are better). As we may observe, {\sc tas-online} consumes most ($\approx 90 \%$) of its available budget, at the same consumption level as the {\sc online greedy}, {\sc random} and {\sc random egoistic} algorithms. The {\sc random egoistic filter} and {\sc tas-offline} algorithms seem to make a slightly better usage of their budget. Nevertheless, the extra cost consumed by {\sc tas-online} is small, especially as expertise levels grow and more experts need to be paid (i.e. for $x>1.2$). The significance of this extra cost gets even smaller considering what we gain in terms of the objective function (Figure~\ref{fig:scal_obj_function1}), where {\sc tas-online} consistently accomplishes more jobs (almost up to double for $x=1.2$) than {\sc random egoistic filter}. As such, {\sc tas-online} has a much 
higher `value-for-mone' (jobs done vs. cost ratio) compared to its competitors. 

\textbf{Effect on Flow Time}. Figure~\ref{fig:flow_vs_exp} shows the flow time of the algorithms for varying levels of expertise availability, and smaller $y$ axis values are better. As we may observe, the proposed {\sc tas-online} algorithm behaves similarly to the rest of the online algorithms. This shows that there is no trade-off of performance for time, i.e.~our algorithm does not achieve its higher objective function values at the cost of flow time.

\textbf{Effect on Number of Assigned Workers.}
Figure~\ref{fig:assign_vs_exp} shows the change in the average number of workers per task, as we modify the availability of expertise, and lower values of the $y$ axis are better. Here, and for most algorithms, we observe a very steep drop in the number of assigned workers, as the average expertise of the crowd worker population increases. This fact is to be expected, as the algorithms need to assign multiple workers to achieve the quality thresholds when expertise is scarce. 
%\textcolor{red}{to explain why random egoistic filter and tas offlline do not behave similarly. For random egoistic filter its probably because it filters out non-qualitative workers?}.

\begin{figure*}
\centering
\begin{minipage}{0.45\textwidth}
\centering
\includegraphics[width=1.0\textwidth]{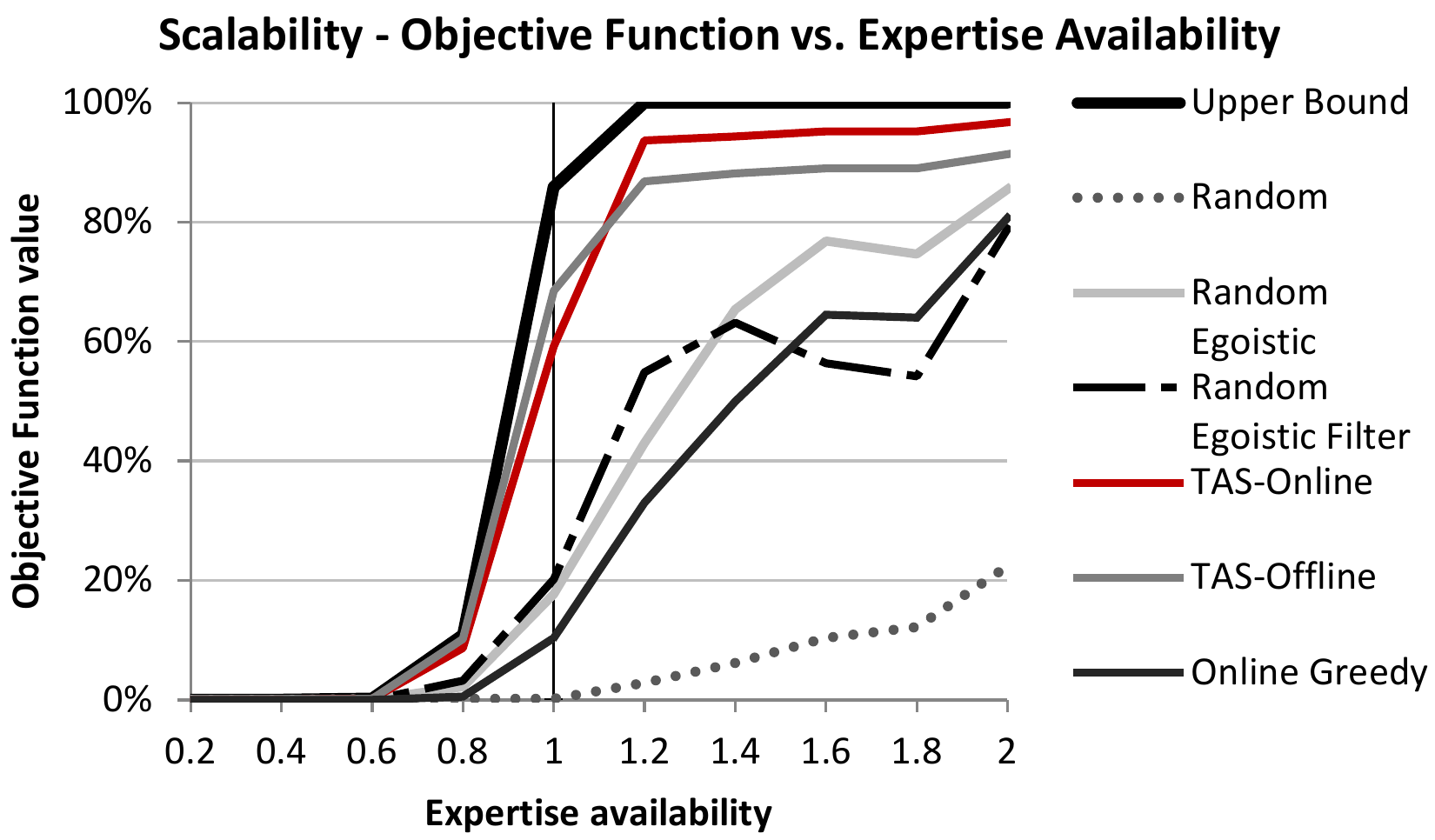}
\caption{Objective function vs. expertise availability. The vertical line corresponds to the baseline simulated instance.}
\label{fig:scal_obj_function1}
\end{minipage}\hfill
\begin{minipage}{0.45\textwidth}
\centering
\includegraphics[width=1.0\textwidth]{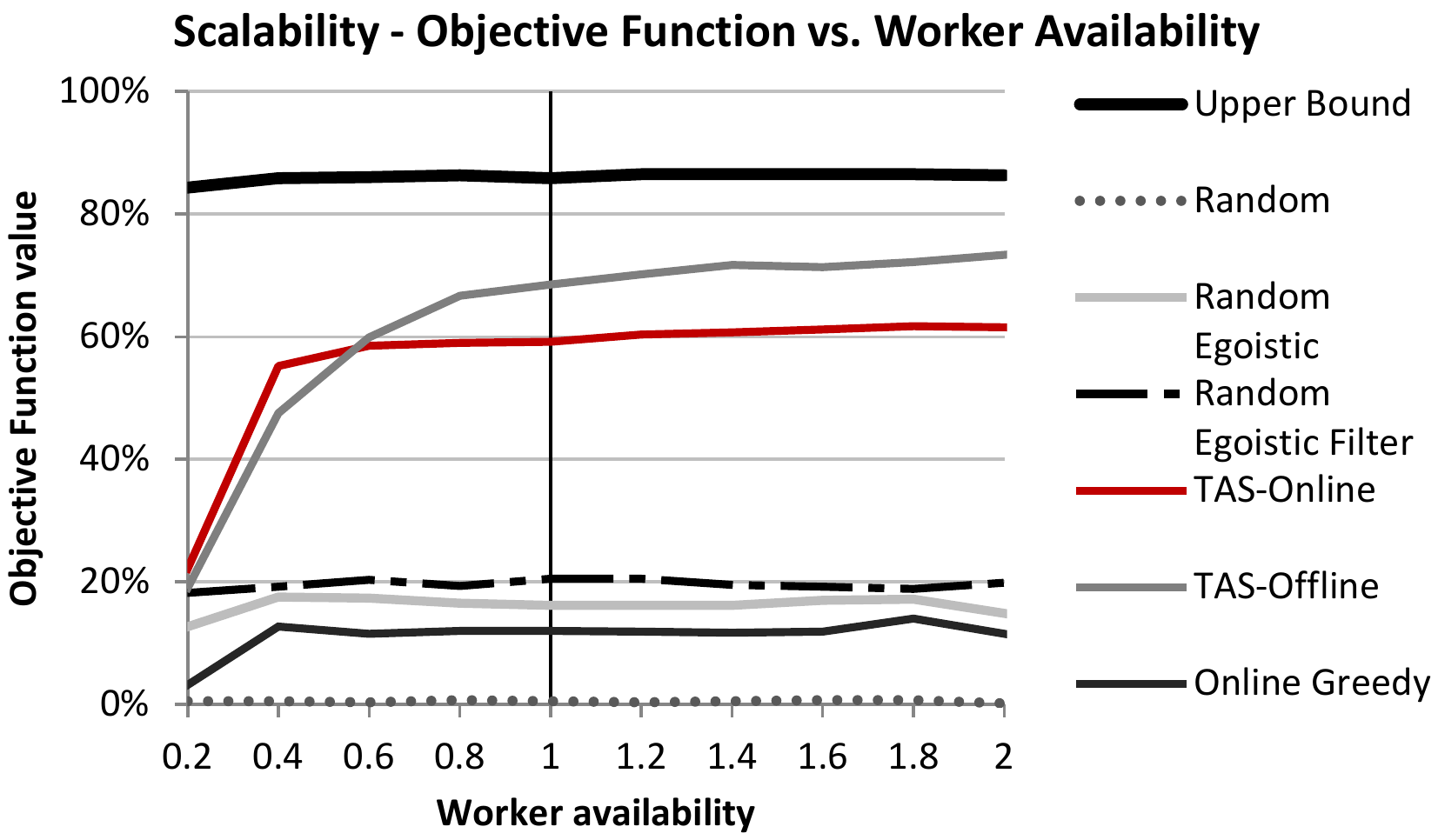}
\caption{Objective function vs. worker availability.}
\label{fig:scal_obj_function2}
\end{minipage}
\end{figure*}

\begin{figure*}
\centering
\begin{minipage}{0.45\textwidth}
\centering
\includegraphics[width=1.0\textwidth]{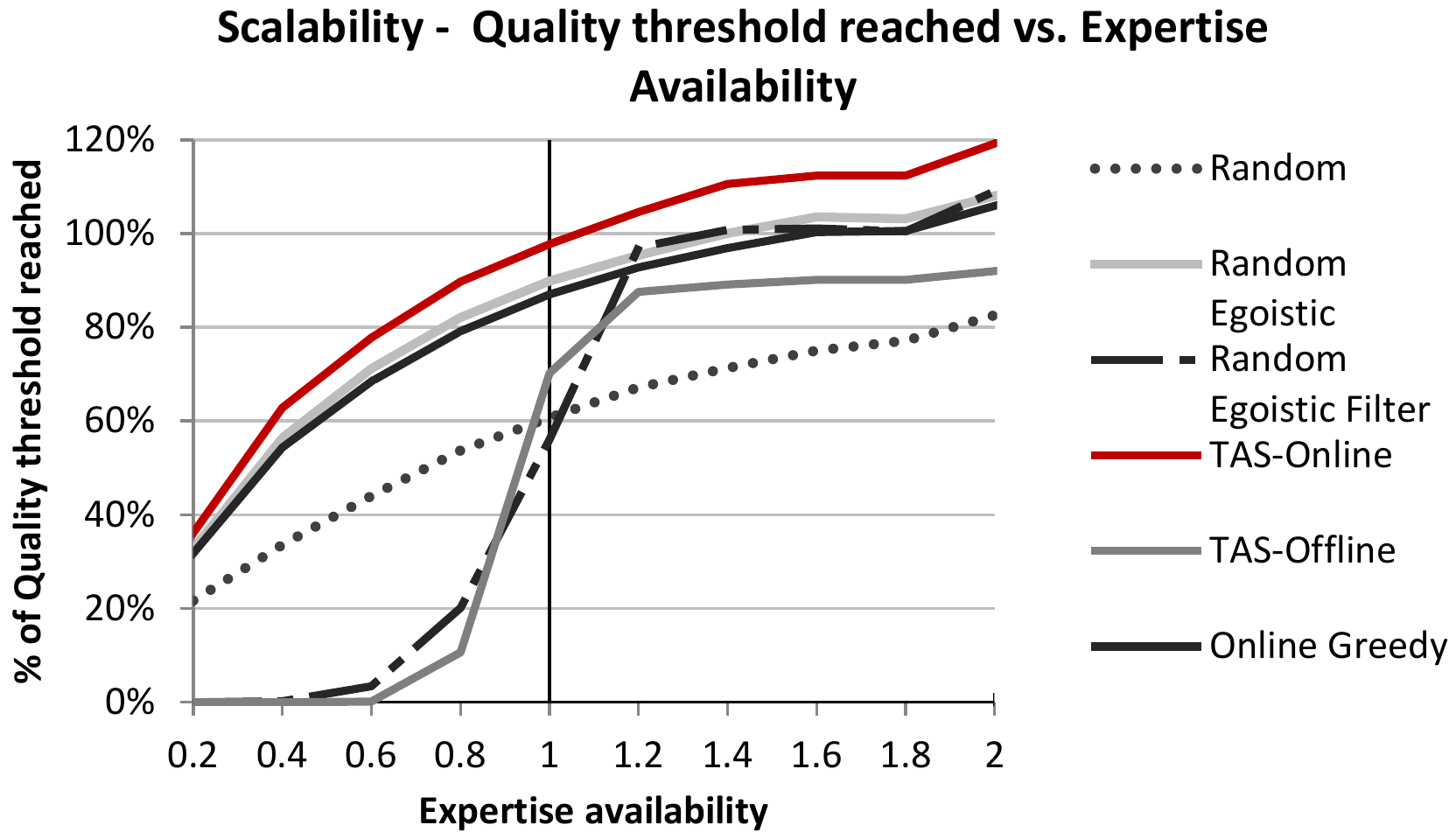}
\caption{Quality reached vs. expertise availability.}
\label{fig:qual_vs_exp}
\end{minipage}\hfill
\begin{minipage}{0.45\textwidth}
\centering
\includegraphics[width=1.0\textwidth]{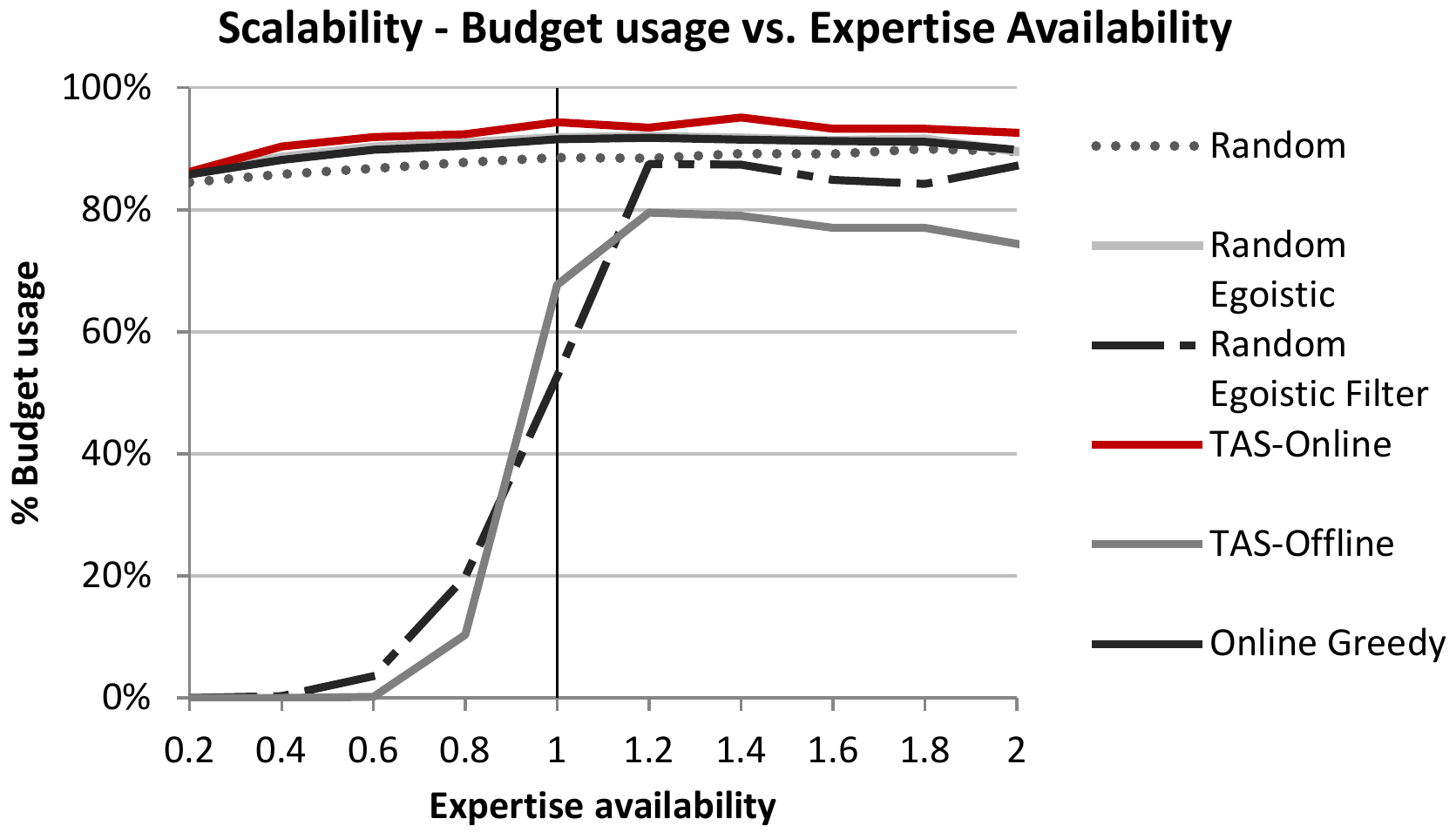}
\caption{Budget availability vs. expertise availability.}
\label{fig:budg_vs_exp}
\end{minipage}
\end{figure*}

\begin{figure*}
\centering
\begin{minipage}{0.45\textwidth}
\centering
\includegraphics[width=1.0\textwidth]{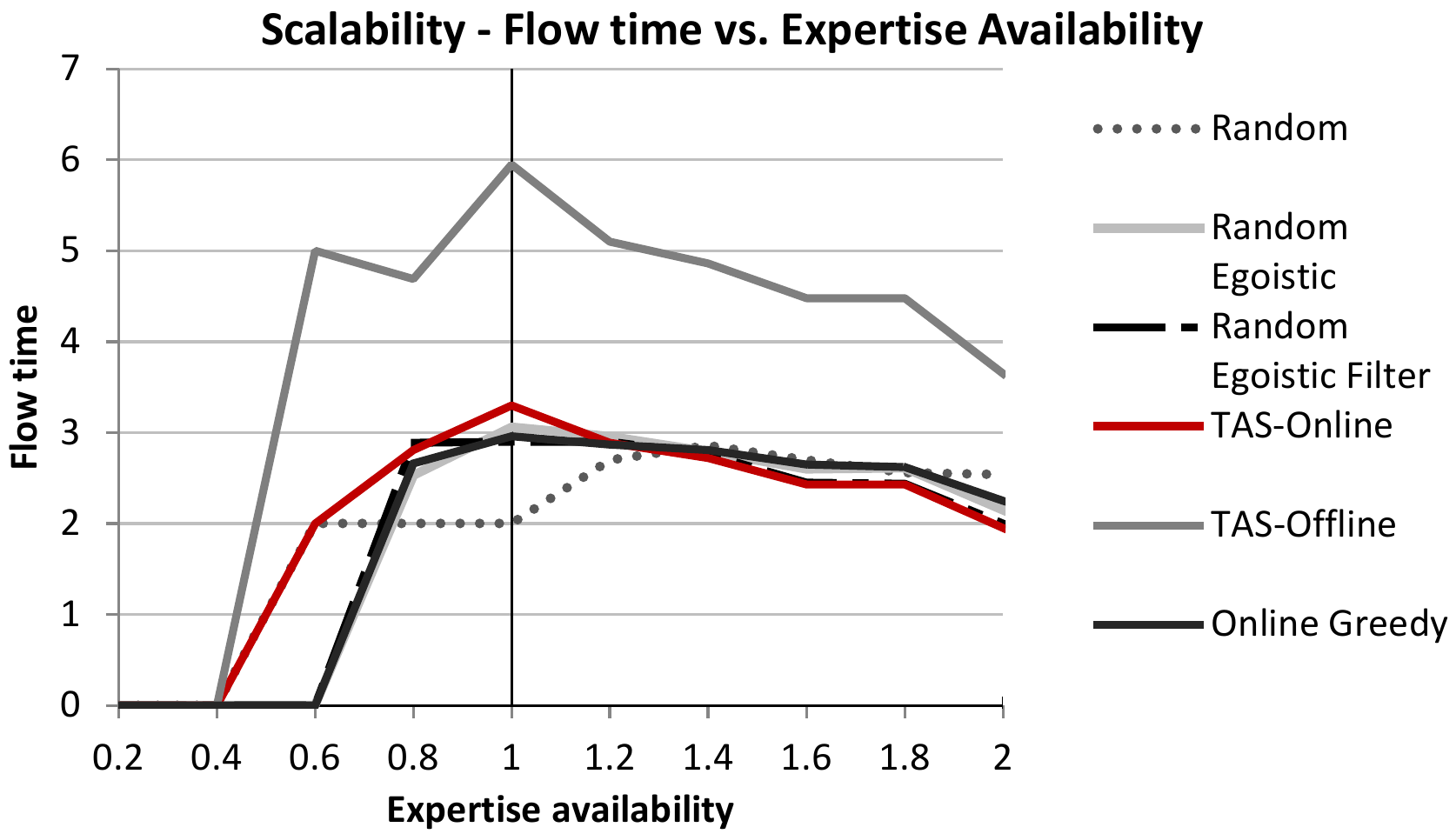}
\caption{Flow time vs. expertise availability.}
\label{fig:flow_vs_exp}
\end{minipage}\hfill
\begin{minipage}{0.45\textwidth}
\centering
\includegraphics[width=1.0\textwidth]{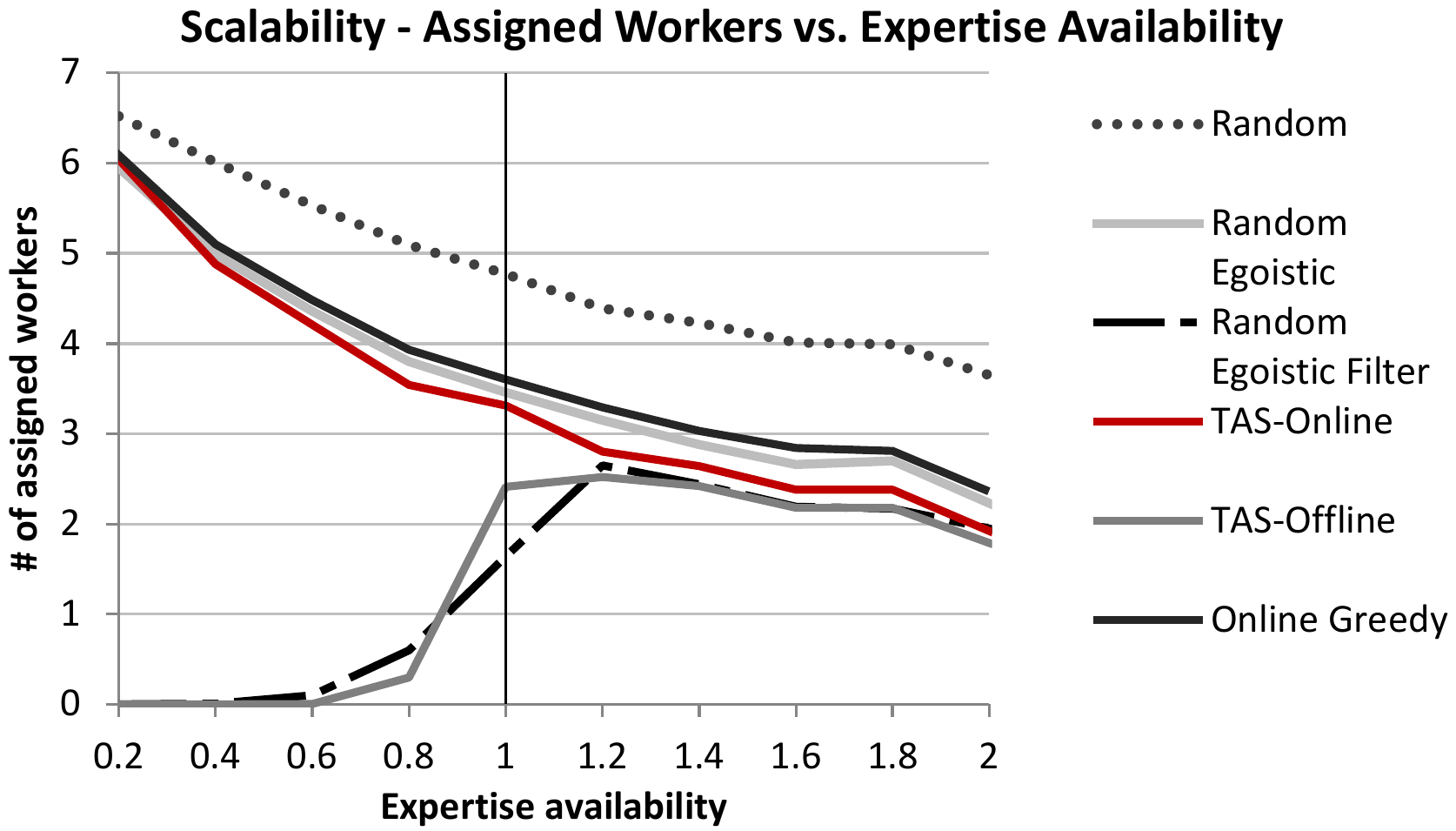}
\caption{Number of assigned workers vs. expertise availability.}
\label{fig:assign_vs_exp}
\end{minipage}
\end{figure*}

\subsection{Real Data Experiment}
\label{realexp}
To examine the effectiveness of the proposed algorithm, we conducted a real world experiment. The platform we used for this was CrowdFlower\footnote{http://www.crowdflower.com}. The task we used was collaborative news article writing, where workers from an initial hiring pool were asked to build on each other's content sequentially, enriching a news article text on a given topic. In more detail, our experimental workflow consisted of three steps. 

In the first step we recruited a pool of 60 workers and recorded their expertise, wage and availability. To measure expertise, we asked each worker to complete two short multiple choice tests, each comprising 10 questions and measuring the workers' know\-ledge skills on a particular topic of current interest, i.e. "The FIFA 2015 corruption scandal" and "Self-driving cars" respectively. Each one of these topics is considered a knowledge domain for the purposes of our experiment. We also gave workers an estimation of the effort that they would have to spend on the second round of the task and asked them to provide us with their required wage. 
% OUTTAKE:
% Finally, we asked each worker to fill in their availability to work on the second step of the experiment, for the next 8 days. 

In the next round we split the hired pool of workers randomly into two parts, one to be used by the benchmark and one by the optimization algorithm during scheduling. We also created six Google documents for each algorithm, three per knowledge domain, which corresponded to the jobs that would have to be accomplished. The quality and cost thresholds, as well as the release date for each job were set according to the same job generation criteria used in the synthetic experiments, and they were the same for the jobs of the benchmark and the optimization algorithm. 
%Job	Release day	Release date	Quality Threshold	Total budget			
%B_FA1	1	26-06-14	60	85
%B_FA2	2	27-06-14	70	100
%B_FA3	3	28-06-15	80	108
%B_SDC1	1	26-06-14	60	85
%B_SDC2	2	27-06-14	70	100
%B_SDC3	3	28-06-15	90	90
%A_FA1	1	26-06-14	60	85
%A_FA2	2	27-06-14	70	100
%A_FA3	3	28-06-15	80	108
%A_SDC1	1	26-06-14	60	85
%A_SDC2	2	27-06-14	70	100
%A_SDC3	3	28-06-15	90	90
In regards to the algorithms to be compared we used {\sc random ego. filter}  as benchmark with $\mathit{factor}=0.3$ (a worker was to be allowed to take a job only if his expertise was at least $30\%$ of the target job quality) and 
{\sc tas-online} as the optimization algorithm. Finally we set the scheduling period to $t=8$ slots and the time unit to one day. 
Each day, one worker would be invited to contribute to each Google document, according to the benchmark and the optimized algorithm. At the end of that day the document would be locked for the particular worker and sent for evaluation by a crowd of $50$ independent crowd workers (different than those used in the experiments) to evaluate the job's current quality. Then, if the job had not surpassed its quality threshold and not exhausted its budget, a new worker was invited to work on the document. 

At the end of the scheduling period, the results were as follows: The benchmark algorithm achieved a successful completion of $3$ out of $6$ jobs, while the optimization algorithm achieved successful completion of $5$ out of $6$ jobs. As it was expected the benchmark algorithm either allowed workers of the minimum necessary expertise to take a job, thus delaying the jobs quality progress too much, or it starved the job of budget. On the other hand {\sc tas-online} selected workers in such a way as to improve job completion within the given time period. As illustrated in Figures~\ref{fig:quality_rate_real} and~\ref{fig:completion_rate_real}, similarly to the respective results of the synthetic experiments, 
{\sc tas-online} achieved higher average job quality and job completion percentages than the benchmark.

\begin{figure}[h]
\center
\includegraphics[scale=0.48]{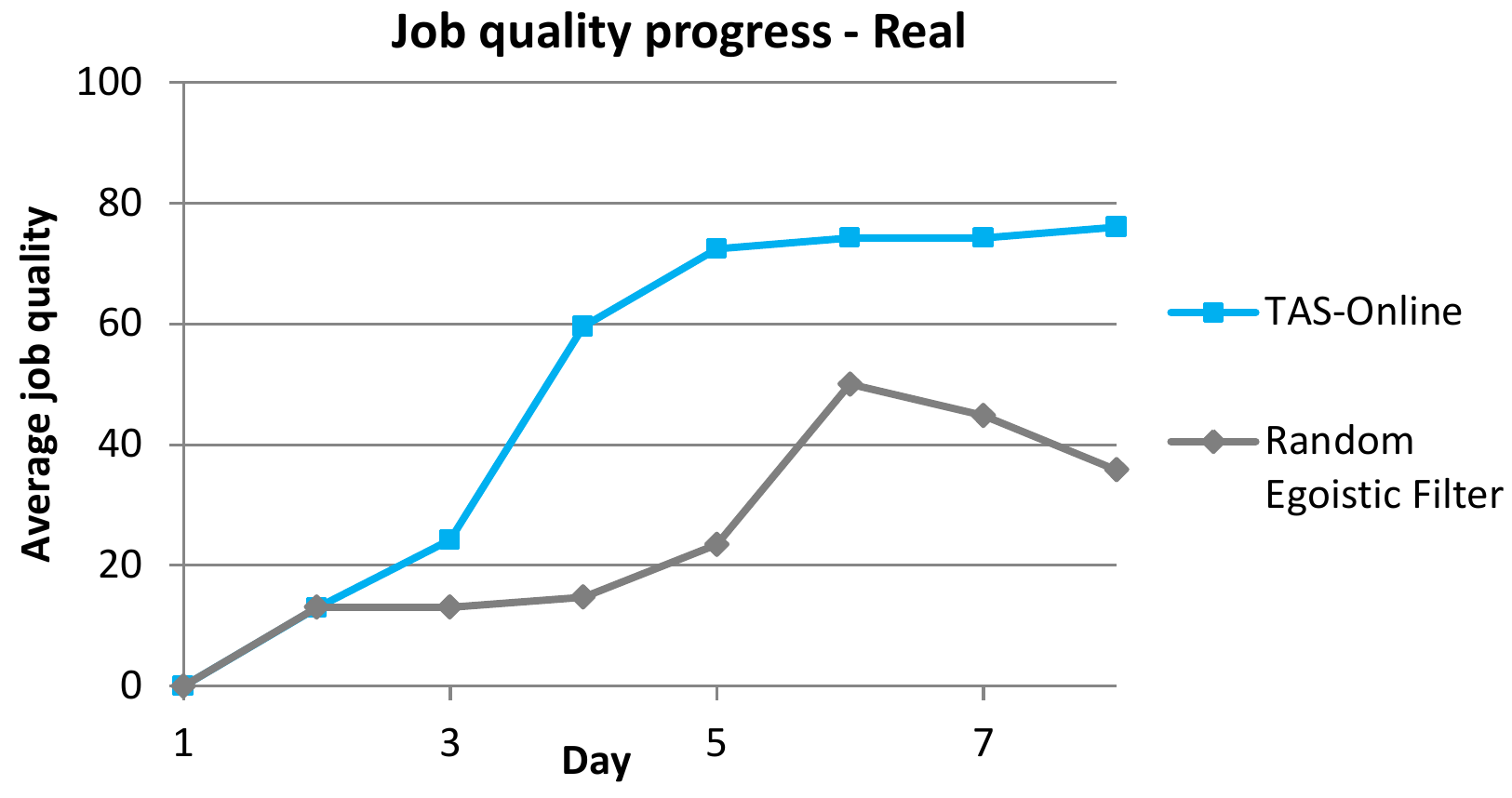}
\caption{Average job quality over time. The proposed  algorithm {\sc tas-online} manages to achieve higher quality per time unit compared to the benchmark. Time unit expressed in days.}
\label{fig:quality_rate_real}
\end{figure}

\begin{figure}[h]
\center
\includegraphics[scale=0.48]{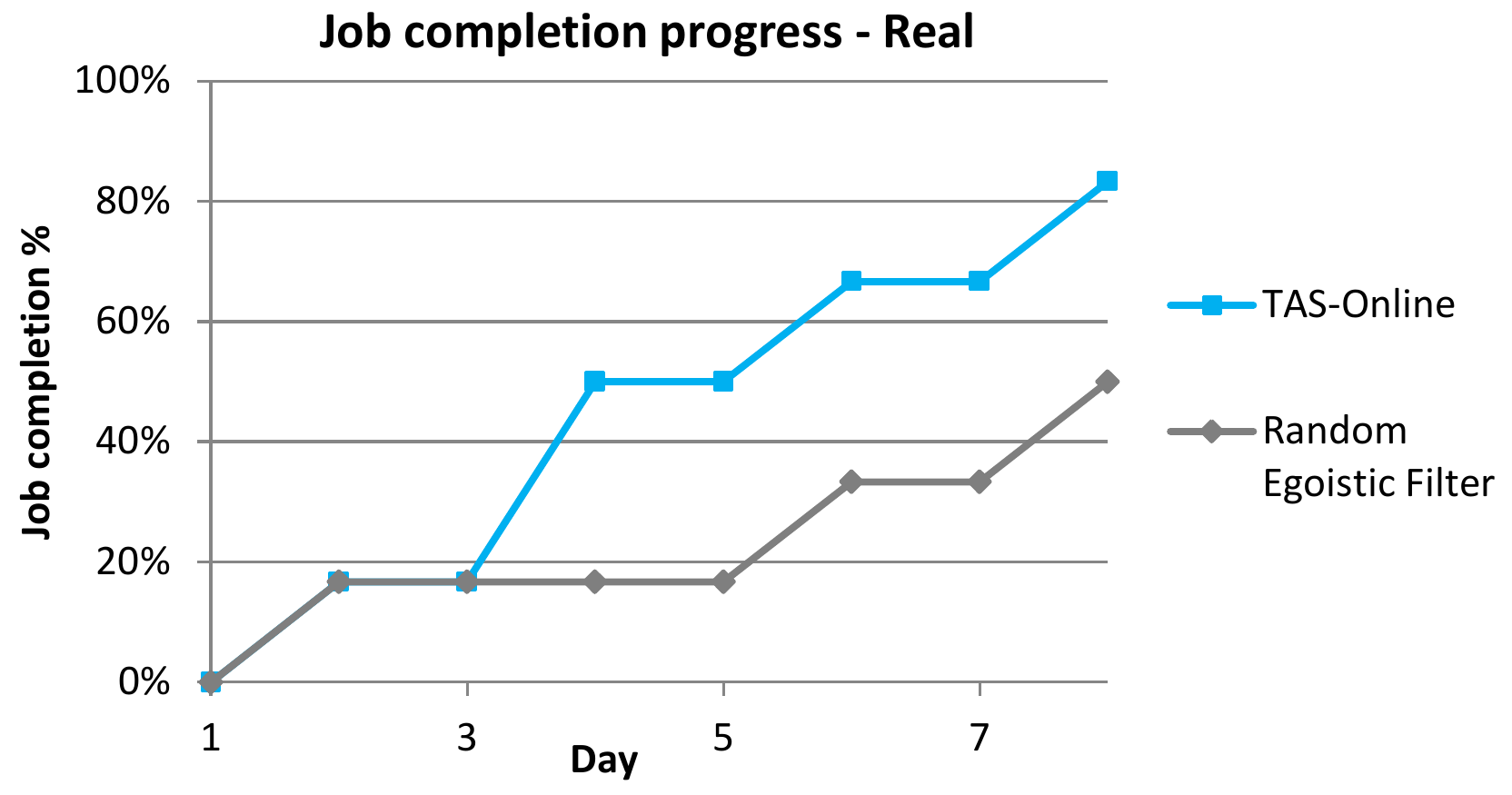}
\caption{Job completion percentage over time. The proposed  algorithm {\sc tas-online} manages to achieve higher job completion rates per time unit compared to the benchmark. Time unit expressed in days.}
\label{fig:completion_rate_real}
\end{figure}
%%%%%%%%%%%%%%%%%%%%%%%%%%%%%%%%%%%%%%%%%%%%%%%%%%%%%%%%%%%%%%%%%%%%%%%%
%%%%%%%%%%%%%%%%%%%%%%%%%%%%%%%%%%%%%%%%%%%%%%%%%%%%%%%%%%%%%%%%%%%%%%%%
\section{Discussion and Future Extensions}
\label{extensions}
The \kic model presented in this paper is the first concrete attempt to incorporate time-sensitive optimization in expert crowdsourcing. Our results, as presented in the previous, indicate that this model can improve the performance of crowdsourcing systems and help them utilize their human capital more effectively. Nonetheless, several challenges still lie ahead and many further extensions can be envisioned. In this final section we briefly discuss how the proposed \kic model and the {\sc tas-online} algorithm can be adapted to address further challenging settings that may appear in practice.

\paragraph{Budget flexibility.} Our initial \kic model assumes a fixed  budget per job, set by the customers before the job is launched. However certain commercial platforms, like CrowdFlower allow jobs to go above their initial budget, and this option could be used in order to recruit better qualified workers during the scheduling period. One simple approach for adapting the {\sc tas-online} algorithm to the `flexible budget option' is to recompute the daily matchings with alternating remaining budgets and explore their effects. A second adaptation is allowing edges to stay in the bipartite graph if the cost exceeds the remaining budget by a given fixed percentage, which can even be specified per job. Since we want to avoid that the algorithm makes too much use of the additional budgets we can reduce the profit of such edges accordingly. A multi-objective view of the optimization problem can also be useful to reveal these budget trade-offs.
 
\paragraph{Non-Acceptance of Assignment.} The initial \kic model assumes that worker availability does not change, once the workers declare themselves available for a specific day. An adapted version of this model could be that workers can decline a certain assignment or they may be marked as unavailable by the system after a certain period, waiting for their reaction, has timed-out. A natural adaptation of the {\sc tas-online} algorithm to this problem is to perform a partial recomputation of the matching for the specific day, where all accepted assignments (workers and corresponding tasks) and  
all stalled workers are removed from the graph. Note that this also allows to bring in new workers and jobs that became available only very recently within the day.

\paragraph{Job prioritizing.} The current \kic model gives all jobs the same priority. However some jobs may become more urgent than others during the scheduling period, for example in crowdsourcing systems dealing with crisis response~\cite{DBLP:journals/corr/ImranLC13}. Job prioritization can be incorporated in the model based on a given deadline, on the quality left to reach the threshold (jobs close to threshold go first), or other criteria. Our {\sc tas-online} algorithm uses a flexible profit function to rate all feasible assignments per day.  Therefore if some jobs (or workers) become preferred over others, the profits on the respective edges can be easily changed. This change is possible for individual worker-task combinations and it can additionally be adjusted every day for a close progress control. 

\paragraph{Quality aggregation model.} Like many relevant studies in the area (e.g.~\cite{Goel:2014:ATW:2567948.2577311,roy2015}), our current \kic model uses a sum function to calculate job quality, assuming that a job's quality is the sum of expertise of the individual workers that have contributed to the job. Nevertheless, and depending on the context, this calculation mode may not always reflect correctly a job's quality. For instance in cases of highly subjective tasks the quality of contributions of the same worker may vary even on tasks belonging to the same topic. In these cases, quality may need to be calculated in a different way, e.g. using crowd-based evaluations after each worker contribution, similarly to our real-world experiment. Here we need to distinguish two cases of adaptation for our algorithm. The first is incorporating a profit function that tries to forecast the benefit from a certain assignment, prior to that assignment. The second is incorporating an extra mechanism that is responsible for determining the quality reached per job, day and assignment. The extra costs associated with this mechanism can also form part of the profit function. Both these adaptation are feasible depending on the exact aggregation model needed.
% \textcolor{red}{I think that the above paragraph necessitates an extra pass.}

\paragraph{Learning.} In our initial \kic model we consider expertise as an inherent, fixed property of each worker. For certain tasks however, like creative ones, the expertise of a certain individual can develop over time, and with the number of accomplished tasks, as workers `learn by doing'~\cite{Dow:2012:SCY:2145204.2145355,crowdsnotdrones}. An improved version of the model could recognize this fact and perform an adjustment of worker expertise over time. According to this version, the expertise per worker needed by the {\sc tas-online} during graph construction each day could be the outcome of a previous learning process (e.g. machine learning as in~\cite{Lykourentzou201018}). The online version of the \kic algorithm is particularly well suited for such an dynamic adjustment.

\paragraph{Order of workers per job.} In our problem formulation the order in which the assigned workers per job are placed on the timeline is arbitrary (openshop model), to simulate the first-come first-served mode of functionality of typical crowdsourcing platforms. It could nonetheless be reasonable for certain applications to require a specific order of workers, e.g. in a decreasing order of expertise. 
%\textcolor{red}{Heinz your text actually said: ``non-decresing" instead, which one is correct?}. 
A straightforward approach to adapt {\sc tas-online} to such a request is to drop all the edges in the bipartite graph for each day such that the only workers that remain assignable are those that correspond to the order criteria. Again, this can be adopted over time such that, e.g., each job begins with an assignment of some workers with sufficient but relatively small expertise, to leave room and budget for enhancement.
  
\paragraph{Multiple assignments per worker and timeslot.} Constraint \ref{c1} of our examined \kic model allows only one task per worker and timeslot. It may be reasonable to relax this to some bounded number of tasks per worker and timeslot, which can also be changed over time, for instance to allow multiple assignments to more experienced workers. Currently {\sc tas-online} relies on computing daily matchings, and the matching condition on the worker-side of the bipartite graph corresponds directly to constraint \ref{c1}. However by taking a closer look at how the weighted matching are computed in terms of min-cost flow networks, we can observe that the capacities on source-edges and sink-edges in the flow network are usually set to $1$ to enforce the matching conditions. If we now allow larger values as capacities on sink edges we can also compute worker-task assignments for the relaxed condition of multiple assignments. 

The above correspond to the main modifications that could be made to adapt the proposed \kic model and {\sc tas-online} algorithm to multiple real-life situations, depending on the crowdsourcing platform, population and type of jobs at-hand. As such they could be used independently or in various combinations, as the starting points for further studies in this promising new field of expert crowdsourcing optimization.
%%%%%%%%%%%%%%%%%%%%%%%%%%%%%%%%%%%%%%%%%%%%%%%%%%%%%%%%%%%%%%%%%%%%%%%%
%%%%%%%%%%%%%%%%%%%%%%%%%%%%%%%%%%%%%%%%%%%%%%%%%%%%%%%%%%%%%%%%%%%%%%%%

\section{Conclusion}
\label{conclusion}

In this paper we present \kic, a model that adds the timeline and online perspective to task assignment optimization in expert crowdsourcing. Using a greedy scheduling algorithm, {\sc tas-online}, we show that optimization under this model can significantly improve expert crowdsourcing performance. Future extensions to enable the \kic model to handle further challenging settings include: 
include adding budget flexibility, job and/or worker prioritizing, fine-tuning the job quality aggregation mechanism, worker performance variability over time, and multiple assignments per worker/timeslot.

\section{Acknowledgments}
The work of Ioanna Lykourentzou in this paper has received funding support by the Luxembourg National Research Fund (FNR) under INTER Mobility grant \#8734708. The present paper is an extension of the work first presented at the Third AAAI Conference on Human Computation and Crowdsourcing (HCOMP 2015)~\cite{2015hcomp}.

\bibliographystyle{plain}
\bibliography{paperbib}

\begin{thebibliography}{10}

\bibitem{Anagnostopoulos:2012:OTF:2187836.2187950}
Aris Anagnostopoulos, Luca Becchetti, Carlos Castillo, Aristides Gionis, and
  Stefano Leonardi.
\newblock Online team formation in social networks.
\newblock In {\em Proceedings of the 21st International Conference on World
  Wide Web}, WWW '12, pages 839--848, New York, NY, USA, 2012. ACM.

\bibitem{roy2015}
Senjuti Basu~Roy, Ioanna Lykourentzou, Saravanan Thirumuruganathan, Sihem
  Amer-Yahia, and Gautam Das.
\newblock Task assignment optimization in knowledge-intensive crowdsourcing.
\newblock {\em The VLDB Journal}, pages 1--25, 2015.

\bibitem{Bernstein:2011:CTS:2047196.2047201}
Michael~S. Bernstein, Joel Brandt, Robert~C. Miller, and David~R. Karger.
\newblock Crowds in two seconds: Enabling realtime crowd-powered interfaces.
\newblock In {\em Proceedings of the 24th Annual ACM Symposium on User
  Interface Software and Technology}, UIST '11, pages 33--42, New York, NY,
  USA, 2011. ACM.

\bibitem{journals/corr/abs-1204-2995}
Michael~S. Bernstein, David~R. Karger, Robert~C. Miller, and Joel Brandt.
\newblock Analytic methods for optimizing realtime crowdsourcing.
\newblock {\em CoRR}, abs/1204.2995, 2012.

\bibitem{Biswas:2015:TBF:2772879.2773291}
Arpita Biswas, Shweta Jain, Debmalya Mandal, and Y.~Narahari.
\newblock A truthful budget feasible multi-armed bandit mechanism for
  crowdsourcing time critical tasks.
\newblock In {\em Proceedings of the 2015 International Conference on
  Autonomous Agents and Multiagent Systems}, AAMAS '15, pages 1101--1109,
  Richland, SC, 2015. International Foundation for Autonomous Agents and
  Multiagent Systems.

\bibitem{Borodin:2005wx}
Allan Borodin and Ran El-Yaniv.
\newblock {\em {Online Computation and Competitive Analysis}}.
\newblock Cambridge University Press, 2005.

\bibitem{6888877}
I.~Boutsis and V.~Kalogeraki.
\newblock On task assignment for real-time reliable crowdsourcing.
\newblock In {\em Distributed Computing Systems (ICDCS), 2014 IEEE 34th
  International Conference on}, pages 1--10, June 2014.

\bibitem{Dalip:2011:AAD:2063504.2063507}
Daniel~Hasan Dalip, Marcos~Andr{\'e} Gon\c{c}alves, Marco Cristo, and P\'{a}vel
  Calado.
\newblock Automatic assessment of document quality in web collaborative digital
  libraries.
\newblock {\em J. Data and Information Quality}, 2(3):14:1--14:30, December
  2011.

\bibitem{Dow:2012:SCY:2145204.2145355}
Steven Dow, Anand Kulkarni, Scott Klemmer, and Bj\"{o}rn Hartmann.
\newblock Shepherding the crowd yields better work.
\newblock In {\em Proceedings of the ACM 2012 Conference on Computer Supported
  Cooperative Work}, CSCW '12, pages 1013--1022, New York, NY, USA, 2012. ACM.

\bibitem{Downs:2010:YPG:1753326.1753688}
Julie~S. Downs, Mandy~B. Holbrook, Steve Sheng, and Lorrie~Faith Cranor.
\newblock Are your participants gaming the system?: Screening mechanical turk
  workers.
\newblock In {\em Proceedings of the SIGCHI Conference on Human Factors in
  Computing Systems}, CHI '10, pages 2399--2402, New York, NY, USA, 2010. ACM.

\bibitem{4567876}
S~Even, A~Itai, and A~Shamir.
\newblock {On the complexity of time table and multi-commodity flow problems}.
\newblock In {\em Foundations of Computer Science, 1975., 16th Annual Symposium
  on}, pages 184--193, October 1975.

\bibitem{7052378}
J.~Fan, M.~Zhang, S.~Kok, M.~Lu, and B.~Ooi.
\newblock Crowdop: Query optimization for declarative crowdsourcing systems.
\newblock {\em Knowledge and Data Engineering, IEEE Transactions on},
  PP(99):1--1, 2015.

\bibitem{conf/aaai/FaradaniHI11}
Siamak Faradani, Bjoern Hartmann, and Panagiotis~G. Ipeirotis.
\newblock What's the right price? pricing tasks for finishing on time.
\newblock In {\em Human Computation}, volume WS-11-11 of {\em AAAI Workshops}.
  AAAI, 2011.

\bibitem{Goel:2014:ATW:2567948.2577311}
Gagan Goel, Afshin Nikzad, and Adish Singla.
\newblock Allocating tasks to workers with matching constraints: Truthful
  mechanisms for crowdsourcing markets.
\newblock In {\em Proceedings of the Companion Publication of the 23rd
  International Conference on World Wide Web Companion}, WWW Companion '14,
  pages 279--280, Republic and Canton of Geneva, Switzerland, 2014.
  International World Wide Web Conferences Steering Committee.

\bibitem{Haas:2015:AMC:2824032.2824062}
Daniel Haas, Jason Ansel, Lydia Gu, and Adam Marcus.
\newblock Argonaut: Macrotask crowdsourcing for complex data processing.
\newblock {\em Proc. VLDB Endow.}, 8(12):1642--1653, August 2015.

\bibitem{DBLP:conf/aaai/HoV12}
Chien{-}Ju Ho and Jennifer~Wortman Vaughan.
\newblock Online task assignment in crowdsourcing markets.
\newblock In {\em Proceedings of the Twenty-Sixth {AAAI} Conference on
  Artificial Intelligence, July 22-26, 2012, Toronto, Ontario, Canada.}, 2012.

\bibitem{Horton:2010:LEP:1807342.1807376}
John~Joseph Horton and Lydia~B. Chilton.
\newblock The labor economics of paid crowdsourcing.
\newblock In {\em EC '10}, 2010.

\bibitem{DBLP:journals/corr/ImranLC13}
Muhammad Imran, Ioanna Lykourentzou, and Carlos Castillo.
\newblock Engineering crowdsourced stream processing systems.
\newblock {\em CoRR}, abs/1310.5463, 2013.

\bibitem{Ipeirotis:2014:QTC:2566486.2567988}
Panagiotis~G. Ipeirotis and Evgeniy Gabrilovich.
\newblock Quizz: Targeted crowdsourcing with a billion (potential) users.
\newblock In {\em Proceedings of the 23rd International Conference on World
  Wide Web}, WWW '14, pages 143--154, New York, NY, USA, 2014. ACM.

\bibitem{jabbari2015}
Hsu~J. Jabbari~S., Assadi~S.
\newblock Online assignment of heterogeneous tasks in crowdsourcing markets.
\newblock In {\em Third AAAI Conference on Human Computation and Crowdsourcing
  (HCOMP-2015)}, 2015.

\bibitem{Josang:2007:STR:1225318.1225716}
Audun J{\o}sang, Roslan Ismail, and Colin Boyd.
\newblock A survey of trust and reputation systems for online service
  provision.
\newblock {\em Decis. Support Syst.}, 43(2):618--644, March 2007.

\bibitem{DBLP:conf/hcomp/KaranamCCDR14}
Saraschandra Karanam, Deepthi Chander, L.~Elisa Celis, Koustuv Dasgupta, and
  Vaibhav Rajan.
\newblock Adaptive performance optimization over crowd labor channels.
\newblock In {\em Proceedings of the Seconf {AAAI} Conference on Human
  Computation and Crowdsourcing, {HCOMP} 2014, November 2-4, 2014, Pittsburgh,
  Pennsylvania, {USA}}, 2014.

\bibitem{kargerBudget}
David~R. Karger, Sewoong Oh, and Devavrat Shah.
\newblock Budget-optimal task allocation for reliable crowdsourcing systems.
\newblock {\em CoRR}, abs/1110.3564, 2011.

\bibitem{Karp:1972tu}
R~M Karp.
\newblock {Reducibility among combinatorial problems}.
\newblock {\em Complexity of Computer Computations}, January 1972.

\bibitem{Kellerer:2013wj}
Hans Kellerer, Ulrich Pferschy, and David Pisinger.
\newblock {\em {Knapsack Problems}}.
\newblock Springer Science {\&} Business Media, March 2013.

\bibitem{6450935}
R.~Khazankin, B.~Satzger, and S.~Dustdar.
\newblock Optimized execution of business processes on crowdsourcing platforms.
\newblock In {\em Collaborative Computing: Networking, Applications and
  Worksharing (CollaborateCom), 2012 8th International Conference on}, pages
  443--451, Oct 2012.

\bibitem{Kravchenko:2000vu}
Svetlana~A Kravchenko.
\newblock {On the complexity of minimizing the number of late jobs in unit time
  open shop}.
\newblock {\em Discrete Applied Mathematics}, 100(1-2), March 2000.

\bibitem{Kuhn:1955to}
H~W Kuhn.
\newblock {The Hungarian method for the assignment problem}.
\newblock {\em Naval research logistics quarterly}, 1955.

\bibitem{Lykourentzou201018}
Ioanna Lykourentzou, Katerina Papadaki, Dimitrios~J. Vergados, Despina Polemi,
  and Vassili Loumos.
\newblock Corpwiki: A self-regulating wiki to promote corporate collective
  intelligence through expert peer matching.
\newblock {\em Information Sciences}, 180(1):18 -- 38, 2010.
\newblock Special Issue on Collective Intelligence.

\bibitem{2015hcomp}
Ioanna Lykourentzou and Heinz Schmitz.
\newblock An online approach to task assignment and sequencing in expert
  crowdsourcing.
\newblock Third AAAI Conference on Human Computation and Crowdsourcing (HCOMP
  2015), Works-in-Progress, nov 2015.

\bibitem{Lykourentzou:2013:IWA:2584064.2584070}
Ioanna Lykourentzou, Dimitrios~J. Vergados, and Yannick Naudet.
\newblock Improving wiki article quality through crowd coordination: A resource
  allocation approach.
\newblock {\em Int. J. Semant. Web Inf. Syst.}, 9(3):105--125, July 2013.

\bibitem{MarchettiSpaccamela:1995ih}
A~Marchetti-Spaccamela and C~Vercellis.
\newblock {Stochastic on-line knapsack problems}.
\newblock {\em Mathematical Programming}, 68(1-3):73--104, January 1995.

\bibitem{minder2012}
P.~Minder, S.~Seuken, A.~Bernstein, and M.~Zollinger.
\newblock Crowdmanager-combinatorial allocation and pricing of crowdsourcing
  tasks with time constraints.
\newblock In {\em 2nd Workshop on Social Computing and User Generated Content}.
  ACM Conference on Electronic Commerce (ACM-EC 2012), jun 2015.

\bibitem{Mo2013}
Luyi Mo, Reynold Cheng, Ben Kao, Xuan~S. Yang, Chenghui Ren, Siyu Lei, David~W.
  Cheung, and Eric Loz.
\newblock Optimizing plurality for human intelligence tasks.
\newblock In {\em CIKM’13}, CIKM’13, New York, NY, USA, 2013. ACM.

\bibitem{Nath:2012:MDT:2436756.2436773}
Swaprava Nath, Pankaj Dayama, Dinesh Garg, Yadati Narahari, and James Zou.
\newblock Mechanism design for time critical and cost critical task execution
  via crowdsourcing.
\newblock In {\em Proceedings of the 8th International Conference on Internet
  and Network Economics}, WINE'12, pages 212--226, Berlin, Heidelberg, 2012.
  Springer-Verlag.

\bibitem{rajan2013}
V.~Rajan, S.~Bhattacharya, L.~Celis, D.~Chander, K.~Dasgupta, and S.~Karanam.
\newblock Crowdcontrol: An online learning approach for optimal task scheduling
  in a dynamic crowd platform.
\newblock In {\em ICML'13 Workshop: Machine Learning Meets Crowdsourcing}. ACM
  Conference on Electronic Commerce (ACM-EC 2012), jun 2015.

\bibitem{DBLP:conf/icwsm/RogstadiusKKSLV11}
Jakob Rogstadius, Vassilis Kostakos, Aniket Kittur, Boris Smus, Jim Laredo, and
  Maja Vukovic.
\newblock An assessment of intrinsic and extrinsic motivation on task
  performance in crowdsourcing markets.
\newblock In {\em Proceedings of the Fifth International Conference on Weblogs
  and Social Media, Barcelona, Catalonia, Spain, July 17-21, 2011}, 2011.

\bibitem{crowdsnotdrones}
S.B. Roy, I.~Lykourentzou, S.~Thirumuruganathan, S.~Amer-Yahia, and G.~Das.
\newblock Crowds, not drones: modeling human factors in interactive
  crowdsourcing.
\newblock In Reynold Cheng, Anish~Das Sarma, Silviu Maniu, and Pierre
  Senellart, editors, {\em DBCrowd 2013 - VLDB Workshop on Databases and
  Crowdsourcing}, volume CEUR Workshop Proceedings, pages 39--42. CEUR-WS,
  2013.

\bibitem{DBLP:journals/corr/RoyLTAD14}
Senjuti~Basu Roy, Ioanna Lykourentzou, Saravanan Thirumuruganathan, Sihem
  Amer{-}Yahia, and Gautam Das.
\newblock Optimization in knowledge-intensive crowdsourcing.
\newblock {\em CoRR}, abs/1401.1302, 2014.

\bibitem{6782912}
Han Yu, Zhiqi Shen, and C.~Leung.
\newblock Bringing reputation-awareness into crowdsourcing.
\newblock In {\em Information, Communications and Signal Processing (ICICS)
  2013 9th International Conference on}, pages 1--5, Dec 2013.

\bibitem{yue2015}
Tao Yue, Shaukat Ali, and Shuai Wang.
\newblock An evolutionary and automated virtual team making approach for
  crowdsourcing platforms.
\newblock In Wei Li, Michael~N. Huhns, Wei-Tek Tsai, and Wenjun Wu, editors,
  {\em Crowdsourcing}, Progress in IS, pages 113--130. Springer Berlin
  Heidelberg, 2015.

\bibitem{Yuen:2011:TMC:2085036.2085206}
Man-Ching Yuen, Irwin King, and Kwong-Sak Leung.
\newblock Task matching in crowdsourcing.
\newblock In {\em Proceedings of the 2011 International Conference on Internet
  of Things and 4th International Conference on Cyber, Physical and Social
  Computing}, ITHINGSCPSCOM '11, pages 409--412, Washington, DC, USA, 2011.
  IEEE Computer Society.

\bibitem{Yuen:2012:TRC:2442657.2442661}
Man-Ching Yuen, Irwin King, and Kwong-Sak Leung.
\newblock Task recommendation in crowdsourcing systems.
\newblock In {\em Proceedings of the First International Workshop on
  Crowdsourcing and Data Mining}, CrowdKDD '12, pages 22--26, New York, NY,
  USA, 2012. ACM.

\bibitem{Yuen:2012:TPM:2428413.2428477}
Man-Ching Yuen, Irwin King, and Kwong-Sak Leung.
\newblock Taskrec: probabilistic matrix factorization in task recommendation in
  crowdsourcing systems.
\newblock In {\em Proceedings of the 19th international conference on Neural
  Information Processing - Volume Part II}, ICONIP'12, pages 516--525, Berlin,
  Heidelberg, 2012. Springer-Verlag.

\end{thebibliography}

\end{document}